\documentclass[a4paper, USenglish, cleveref]{lipics-v2021}
\usepackage{xspace}
\usepackage{import}
\usepackage[colorinlistoftodos,prependcaption,textsize=scriptsize,disable]{todonotes}

\usepackage[T1]{fontenc}
\usepackage[utf8]{inputenc}
\usepackage{amsmath}
\usepackage{amsfonts}
\usepackage{amsthm}
\usepackage{enumitem}
\usepackage{nth}

\usepackage{algpseudocode}
\usepackage{algorithm}

\usepackage{verbatim}
\usepackage{graphicx}
\graphicspath{ {./pictures/} }
\usepackage[rightcaption]{sidecap}
\usepackage{wrapfig}

\usepackage{chngcntr}
\counterwithout{equation}{section}

\usepackage{algorithm, caption}
\usepackage{algpseudocode}
\algrenewcommand\textproc{}

\newcommand{\OO}[1]{\mathcal{O} \left( #1 \right)}
\newcommand{\OMG}[1]{\Omega \left( #1 \right)}

\newcommand{\Exp}[1]{\mathbb{E} \left[ #1 \right]}
\newcommand{\LExp}[2]{\mathbb{E} \left[ #1 \; \middle| \; #2 \right]}

\newcommand{\ceil}[1]{\left\lceil #1 \right\rceil}
\newcommand{\floor}[1]{\left\lfloor #1 \right\rfloor}
\newcommand{\flfrac}[2]{\floor{\frac{#1}{#2}}}
\newcommand{\cefrac}[2]{\ceil{\frac{#1}{#2}}}

\newcommand{\bra}[1]{\left( #1 \right)}
\newcommand{\set}[1]{ \left\{ #1 \right\} }
\newcommand{\longset}[2]{ \left\{ #1 \; \middle| \; #2 \right\} }
\newcommand{\abs}[1]{\left| #1 \right|}

\newcommand{\term}[1]{\textbf{#1}}

\DeclareMathOperator{\repr}{repr}
\DeclareMathOperator{\rev}{rev}
\DeclareMathOperator{\full}{full}
\DeclareMathOperator{\inv}{inv}
\DeclareMathOperator{\dist}{dist}
\newcommand{\mname}[1]{\mathrm{ #1 }}
\title{Sublogarithmic Approximation for Tollbooth Pricing on a Cactus}
\titlerunning{Tollbooth Pricing on a Cactus} 

\author{Andrzej Turko}{University of Wrocław, Poland}{andrzej.turko@gmail.com}{}{}
\author{Jarosław Byrka}{University of Wrocław, Poland}{jby@cs.uni.wroc.pl}{https://orcid.org/0000-0002-3387-0913}{NCN grant number 2020/39/B/ST6/01641.}

\authorrunning{A. Turko, J. Byrka} 
\Copyright{Andrzej Turko, Jarosław Byrka} 



\nolinenumbers 

\hideLIPIcs  

\EventEditors{John Q. Open and Joan R. Access}
\EventNoEds{2}
\EventLongTitle{42nd Conference on Very Important Topics (CVIT 2016)}
\EventShortTitle{CVIT 2016}
\EventAcronym{CVIT}
\EventYear{2016}
\EventDate{December 24--27, 2016}
\EventLocation{Little Whinging, United Kingdom}
\EventLogo{}
\SeriesVolume{42}
\ArticleNo{23}

\begin{document}

\maketitle

\begin{abstract}

We study an envy-free pricing problem, in which each buyer wishes to buy a~shortest path connecting her individual pair of vertices in a~network owned by a~single vendor.
The vendor sets the prices of individual edges with the aim of maximizing the total revenue generated by all buyers.
Each customer buys a~path as long as its cost does not exceed her individual budget.
In this case, the revenue generated by her equals the sum of prices of edges along this path.
We consider the unlimited supply setting, where each edge can be sold to arbitrarily many customers.
The problem is to find a~price assignment which maximizes vendor's revenue.
A~special case in which the network is a~tree is known under the name of the tollbooth problem.
Gamzu and Segev proposed a~$\OO{\frac{\log m}{\log \log m}}$-approximation algorithm for revenue maximization in that setting. 
Note that paths in a~tree network are unique, and hence the tollbooth problem falls under the category of single-minded bidders, i.e., each buyer is interested in a~single fixed set of goods.

In this work we step out of the single-minded setting and consider more general networks that may contain cycles.
We obtain an algorithm for pricing cactus shaped networks, namely networks in which each edge can belong to at most one simple cycle.
Our result is a~polynomial time $\OO{\frac{\log m}{\log \log m}}$-approximation algorithm for revenue maximization in tollbooth pricing on a~cactus graph.
It builds upon the framework of Gamzu and Segev, but requires substantially extending its main ideas: the recursive decomposition of the graph, the dynamic programming for rooted instances and rounding the prices.


\end{abstract}

\newpage

\section{Introduction}
\label{sect:intro}


The problem of maximizing revenue by setting optimal prices has been widely studied in various settings (see, e.g., \cite{largeMarkets, priceDoubling, impreciseDistribution}).
This work discusses the problem of envy-free pricing for revenue maximization.
In general, this problem can be modeled as a two phase game.
In the first step, vendor assigns prices to the offered goods.
Then, each buyer purchases her most preferred subset of goods based on given prices and her own preferences.
Every buyer aims to maximize her utility, and the seller aims to maximize the total price paid by customers.
The problem is to find an optimal strategy for the vendor.

More precisely, an instance of the envy-free pricing problem consists of $m$ goods and $n$ buyers.
Each buyer is defined by a function which assigns a non-negative valuation to every subset of the goods.
It is assumed that the valuation of an empty set for each customer equals zero.
A solution to the problem is formed by non-negative prices of goods and an envy-free allocation of goods to the buyers.
Utility of a buyer from a set of goods equals her valuation of this set minus the total price of its elements.
An allocation is envy-free when no buyer would like to change her assigned set of goods.
In other words, the set assigned to her must maximize her utility.

In this work we focus on the unlimited supply setting, where each one of the $m$ goods can be sold to arbitrarily many buyers.
Such goods may be thought of as intellectual property or access to infrastructure.
Sometimes the limited supply setting is also considered, where each good is available only in a certain number of copies.
In that case, the solution must not only satisfy the envy-freeness constraints, but also the number of buyers any good is allocated to must not exceed its supply.

We study a natural case of the envy-free pricing with unlimited supply, where the goods can be modeled by edges in a graph and buyers wish to purchase cheapest paths.
More precisely, each buyer has equal positive valuations for paths connecting a certain pair of vertices and zero valuation for all the other sets of goods.
Such a problem may be used to model a situation where the vendor is an owner of a road network and buyers are drivers wishing to travel from one city to another.

Guruswami et al. \cite{Guruswami05} have defined and studied two subcases of this scenario called: the tollbooth and the highway problems.
In the former the underlying graph is a tree and in the latter it is a path. We extend this collection by adding cactus graphs that allow edge-disjoint cycles and hence allow more than one path being attractive for a client.
To the best of our knowledge this is the first work that addresses envy-free tollbooth pricing of networks where clients have alternative routes (are not single-minded).

\subsection{Related work}

The problem of envy-free pricing for revenue maximization has been studied in various settings.
We are going to survey mostly the results for single-minded buyers, a model where each buyer has positive valuation for exactly one set of goods.


Guruswami et al. \cite{Guruswami05} defined the single-minded buyers setting and presented a polynomial $\OO{\log m + \log n}$-approximation algorithm for the variant with unlimited supply.
Also for the unlimited supply setting, Balcan, Blum and Mansour \cite{Balcan08} have shown that a logarithmic guarantee on expected revenue can be achieved by randomly setting a single price to all the goods. Notably, this result holds for buyers with arbitrary valuations. By taking it as a reference point, a natural question is: For what valuation classes sublogarithmic approximation of revenue is possible?

Indeed, for special cases of the  unlimited supply setting with single-minded buyers, such results were obtained. For the tollbooth problem, Gamzu and Segev \cite{Gamzu10} achieved a $\OO{\frac{\log m}{\log \log m}}$-approximation of revenue with a polynomial algorithm.
For the highway problem, Grandoni and Rothvoß \cite{grandoni2016pricing} have designed a polynomial time approximation scheme (PTAS).

Already for these two problems hardness results for envy-free pricing with single-minded buyers are known.
Guruswami et al. \cite{Guruswami05} have proven that the tollbooth problem is NP-hard.
This was followed by a result from Briest and Krysta \cite{Briest06}, who showed the same for the highway problem.

For the general envy-free pricing Demaine et al.~\cite{demaine2008combination} showed  several inapproximability results under various complexity assumptions.
They proved a lower bound of $\Omega(\log n)$, under a hardness hypothesis regarding the balanced bipartite independent set problem.
In this context, the result of Gamzu and Segev~\cite{Gamzu10} shows that pricing is strictly simpler to approximate on trees. We extend it to show that sublogarithmic approximation of revenue is also possible on cactus graphs.

Of course, the mentioned impossibility results hold for limited supply as well.
In that setting there also are several approximation results.
Cheung and Swamy \cite{Cheung08} have designed a $\OO{\sqrt{m}\log u_{max}}$-approximation algorithm for the general envy-free pricing problem with single-minded buyers ($u_{max}$ denotes the maximal number of copies of a single good).
In the tollbooth and highway problems they have obtained approximation ratio of $\OO{\log u_{max}}$.
Elbassioni, Fouz and Swamy~\cite{elbassioniWINE10} have obtained a matching approximation guarantee for the non-envy-free tollbooth problem without the single-mindedness constraint.
More recently, Grandoni and Wiese~\cite{grandoni_et_al:LIPIcs:2019:11175} obtained a PTAS for the limited supply version of the highway problem.

\subsection{Our result}

We consider the tollbooth problem on cactus graphs, a natural generalization of the original tollbooth problem (on trees) with unlimited supply.
Instead of requiring the graph to be a tree, we only require that the underlying graph is a cactus, i.e. its every edge belongs to at most one simple cycle. 
The main difference between the two models is that, unlike in a tree, in a cactus there can be multiple simple paths connecting a single pair of vertices.
Thus, each buyer can be interested in purchasing multiple sets of goods, i.e. is not single-minded.   
We obtain the following result:
\begin{theorem}
\label{theorem:main}
There exists a polynomial time approximation algorithm for the tollbooth problem on cactus graphs with unlimited supply which achieves an approximation guarantee for revenue of $\OO{\frac{\log m}{\log \log m}}$, where $m$ is the number of edges of the graph.
\end{theorem}

Our approximation algorithm utilizes a similar framework as the algorithm by Iftah Gamzu and Danny Segev \cite{Gamzu10} for the classical tollbooth problem (on trees). However, various parts of the algorithmic construction are carefully adapted to handle cycles and the freedom of clients to choose one of two routes on each cycle they have on their way.

To the best of our knowledge, this is the first such result for graphs more general than trees, and hence the first one not restricted to the single-minded bidder case.

\subsection{Model and preliminaries}

Let us consider an instance of the tollbooth problem on cactus graphs with $m$ goods and $n$ buyers.
Its description consists of a simple graph $G$ with $m$ edges such that no edge lies on two simple cycles and a set $B$ of buyers.
Each buyer $i \in B$ is described by a pair of vertices $u_i$ and $v_i$, and her budget $b_i > 0$.
For each subset of edges $S$, her valuation is defined in the following way:
$$ f_i(S) = \begin{cases} b_i, & \text{if } S \text{ consists of edges along a } u_i\text{-}v_i \text{ path} \\
						  0, & \text{otherwise} \end{cases} $$	  
A solution is a real vector $p$ assigning non-negative prices to the edges of $G$.
Let us treat the prices as lengths of edges and let $d_i$ denote the distance between $v_i$ and $u_i$.
If $b_i \geq d_i$, $i$-th buyer purchases all edges along a shortest $u_i$-$v_i$ path.
Otherwise, she buys nothing. 
Such an allocation is envy-free.
Note that if there are many shortest $u_i$-$v_i$ paths, choosing either one does not change the revenue.

In this work we present an algorithm for finding such prices that the above-mentioned way of allocating goods to buyers results in revenue $\OO{\frac{\log m}{\log \log m}}$ times smaller than optimal.

\subsection{Overview of techniques}

Our algorithm follows the classify-and-select paradigm of Gamzu and Segev. Buyers are split into $\OO{\frac{\log m}{\log \log m}}$ subsets, which define separate instances of the problem.
Those subproblems are processed independently and constant factor approximations for each of them are computed.
The final solution is obtained by choosing the one yielding the biggest revenue.
Supply of the goods is unlimited and, thus, the revenue of a solution to a subproblem does not decrease when applied to the initial instance with all the buyers.
This gives an $\OO{\frac{\log m}{\log \log m}}$~approximation of revenue, because the total revenue is at most the sum of the revenues of the subproblems.

For each subproblem the algorithm constructs a subgraph of $G$, called the skeleton, such that all paths desired by a buyer in the given instance enter and leave the skeleton exactly once and in the same vertices.
This way, each such path is split into three parts, one of which is in the skeleton and the other two are not.
Note that, for a constant factor approximation, it suffices to collect revenue either only on the skeleton or only outside the skeleton.
In the former case the revenue is achieved by setting appropriate prices of the skeleton's edges.
In the latter it suffices to focus on groups entering or leaving the skeleton through the same vertex leading to rooted instances.
Due to cycles in the underlying graph, several significant challenges arise in both subproblems.

\textbf{Rooted instances: }
The algorithm by Gamzu and Segev solved rooted instances with a black-box dynamic programming algorithm from \cite{Guruswami05}.
In our case the subgraphs forming rooted instances can contain cycles.
Thus, that dynamic programming, which has been designed for trees, could not be applied verbatim.
In order to define its subproblems, we have generalized the notion of a subtree using the tree-like structure of biconnected components.
Another key element of our solution is a technique which effectively transforms cycles into paths.
It is based on the observation, that, as far as shortest paths from vertices to the root are concerned, one edge of a cycle is always redundant.
By guessing this edge, one can tackle the problem on a cycle as if it was a path.
For an optimization problem it is enough to iterate over all possible choices of this edge and calculate the solutions independently using the dynamic programming for a tree.
We believe that this technique has a wide range of applications in generalizing algorithms for trees to cyclic graphs whose biconnected components have simple structure.
Its usage for our problem is described in Section \ref{subsect:rootedCase}.

\textbf{Dependent subgraphs of the skeleton: }
The classification of buyers is based on a recursive decomposition of the input graph.
The algorithm for the original tollbooth problem in each step splits the tree into several connected subgraphs and processes buyers who wish to buy paths with endpoints belonging to different ones.
Those connected subgraphs can be then processed completely independently because all paths relevant to the next instances are fully contained in the individual subgraphs.
This is not the case in cactus graphs, namely for subgraphs which contain edges lying on the same cycle.
As the cycles can be arbitrarily long, our algorithm may need to divide them when splitting the skeleton into smaller parts.
It turns out that the dependencies between resulting subgraphs regard the cost of paths in the cycle shared between them.
By making assumptions about their costs, the algorithm can isolate the subgraphs and process them independently.
This approach, however, results in multiple solutions for each subgraph based on different assumed costs of individual parts of shared cycles.
While merging those solutions into an approximately optimal global one, the procedure from Section \ref{subsubsect:open} controls the cost of each cycle using dynamic programming inspired by the knapsack problem.
In order to make this possible, we have extended the price rounding techniques, which allowed to relax the assumptions about the costs of shared cycles and to calculate approximate revenue.
Those techniques are described in Section \ref{subsect:pricingStrategies}.

\textbf{Decomposing the graph: }
In the original tollbooth problem it was sufficient to split each subtree at a given level of decomposition into connected subtrees of the right size, which formed the next step of the decomposition.
Our solution for handling the dependencies between subgraphs of the skeleton has been made possible by additional properties ensured by the decomposition.
For example, we ensure that two subgraphs forming the decomposition can share at most one cycle.
Another way of limiting the dependencies between subgraphs, to which $G$ is split, is limiting the number of vertices each subgraph shares with the other ones.
The decomposition used by our algorithm is characterized in full detail by Lemma \ref{lemma:decomposition}.

\textbf{Pricing the segments: }
Segments are edge-disjoint subgraphs of the skeleton.
In each of them there are two vertices, called endpoints, and only they can be shared with other segments.
Thus, one can think of them as generalization of edges.
In one of the subproblems the algorithm fixes the lengths of whole segments and sets the prices of the edges inside a given segment so that the revenue from selling the paths starting inside and ending outside it is maximized.
In the original tollbooth problem the buyers are single-minded, so for each such path the endpoint through which it will leave the segment is fixed.
For a cactus graph it is not the case, a buyer may choose a path passing though any of the two endpoints depending on the prices.

We handle that additional complexity in the following way.
With fixed lengths of segments, it is possible to calculate the maximal amount of money a buyer is able to spend on edges inside the segment where her path begins.
Our algorithm uses this to split the buyers into two categories.
Each buyer who cannot afford to pay half of the segment's length, for fixed prices of its edges, can only purchase paths passing through a single fixed endpoint.
As for the remaining buyers, the vendor can charge half of the segment's cost for edges incident to any of the endpoints and they will be able to pay this much.
The procedure based on this idea is described in Section \ref{subsubsect:acyclicSegments}.

\section{Graph decomposition}
\label{sect:decomp}

Using a recursive decomposition of the cactus graph our algorithm splits the buyer set $B$ into disjoint subsets, which are later processed independently.
Here we define this partition.

\subsection{The tree of biconnected components}
\label{subsect:componentTree}

We begin by discussing the structure of biconnected components of the cactus graph.
Let us fix an arbitrary vertex of $G$, denoted $r_G$, as the root of $G$ for the duration of the whole algorithm.
A~vertex or edge is said to be above another one if it is closer to $r_G$.
Every biconnected component of a cactus is either a single edge or a simple cycle. 
Thus, it has a single \term{topmost vertex} and at most two \term{topmost edges}, which are exactly those adjacent to the topmost vertex.

\begin{definition}
\label{def:associated}
For each cycle in $G$, its two edges closest to $r_G$, i.e. topmost edges, form a \term{pair of associated edges}.
Note that every edge belongs to at most one such pair.
\end{definition}

Although each vertex can be a topmost vertex in arbitrarily many biconnected components, it can belong to at most one without being its topmost vertex.
Furthermore, all vertices except for the root belong to exactly one such biconnected component.
Let us call it the main component of this vertex.
For the root it is a special component, consisting only of itself (a single vertex).
Our algorithm uses the tree of biconnected components rooted in this special component.
Every other component is a child of the main component of its topmost vertex.
Note that such a tree is unique.

\begin{definition}
A \term{subtree graph of a component} $C$ is the graph consisting of all the edges and vertices belonging to any descendant of $C$ (inclusive) in the tree of biconnected components.
The subtree graph of the root component is the whole graph $G$.
\end{definition}

\begin{SCfigure}
\centering
\includegraphics[width=0.55\textwidth]{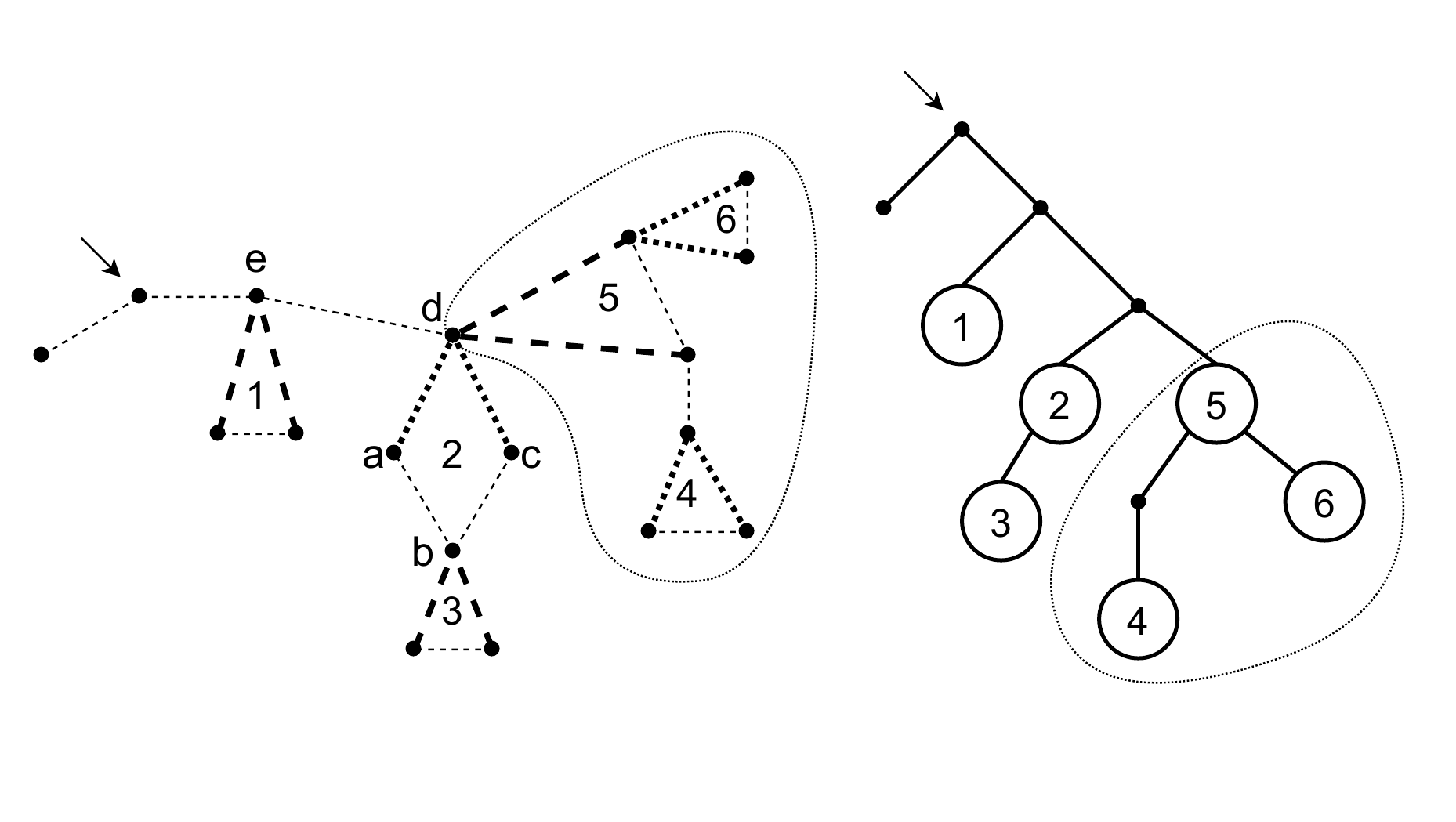}
\caption{
An example cactus graph with marked pairs of associated edges and its tree of biconnected components.
The arrows indicate respectively the root vertex and the root component.
On both drawings the subtree graph of the cycle number 5 is marked out.
Vertices $a$, $b$, $c$ and $d$ belong to beconnected component $2$.
It's also a main component of $a$, $b$ and $c$.
However, the main component of $d$ is the edge $e$-$d$.
Note that $e$ belongs to three distinct biconnected components and it's main component is the edge between $e$ and the root vertex.
}

\end{SCfigure}

\subsection{Balanced decomposition}
\label{subsect:decomposition}

Decompositions $D_1, D_2, \dots D_L$ of $G$ are defined recursively.
Each of them is a family of edge-disjoint subgraphs, called fragments, which cover the graph $G$.
In $D_1$ the whole graph $G$ forms a single fragment, and in $D_L$ each fragment consists of at most two edges.
For all $j < L$ each fragment in partition $D_j$ is split into a number of subgraphs, which become fragments in $D_{j+1}$.

\begin{definition}
A vertex which belongs to multiple fragments in partition $D_{j+1}$ is called a \term{border vertex} of $j$-th level.
Furthermore, every vertex of $G$ is considered to be a \text{border vertex} of $L$-th level.
\end{definition}

Note that a border vertex of $j$-th level is also a border vertex of $(j+1)$-th level.

\begin{lemma}
\label{lemma:decomposition}
Consider a family of decompositions $D_1, D_2, \dots D_L$ of a cactus graph $G$ satisfying the following invariants for each valid $j$:
\begin{enumerate}[nosep]
\item Each fragment in $D_j$ is split into $\OO{k}$ fragments in $D_{j+1}$.
\item The maximal number of edges in a fragment forming $D_{j+1}$ is $\OMG{k}$ times smaller than in $D_j$.
\item Each fragment forming $D_j$ contains at most $\OO{k}$ border vertices of $j$-th level.
\item Each pair of associated edges belongs to the same fragment of $D_j$.
\item All fragments forming $D_j$ are connected subgraphs of $G$.
\end{enumerate}
For $k$ being an unbounded and nondecreasing function of $m$ (the number of edges in $G$), such a family can be found in polynomial time.
\end{lemma}

\begin{wrapfigure}[18]{r}{0.4\textwidth}
\centering
\vspace{-1.15cm}
\includegraphics[width=0.4\textwidth]{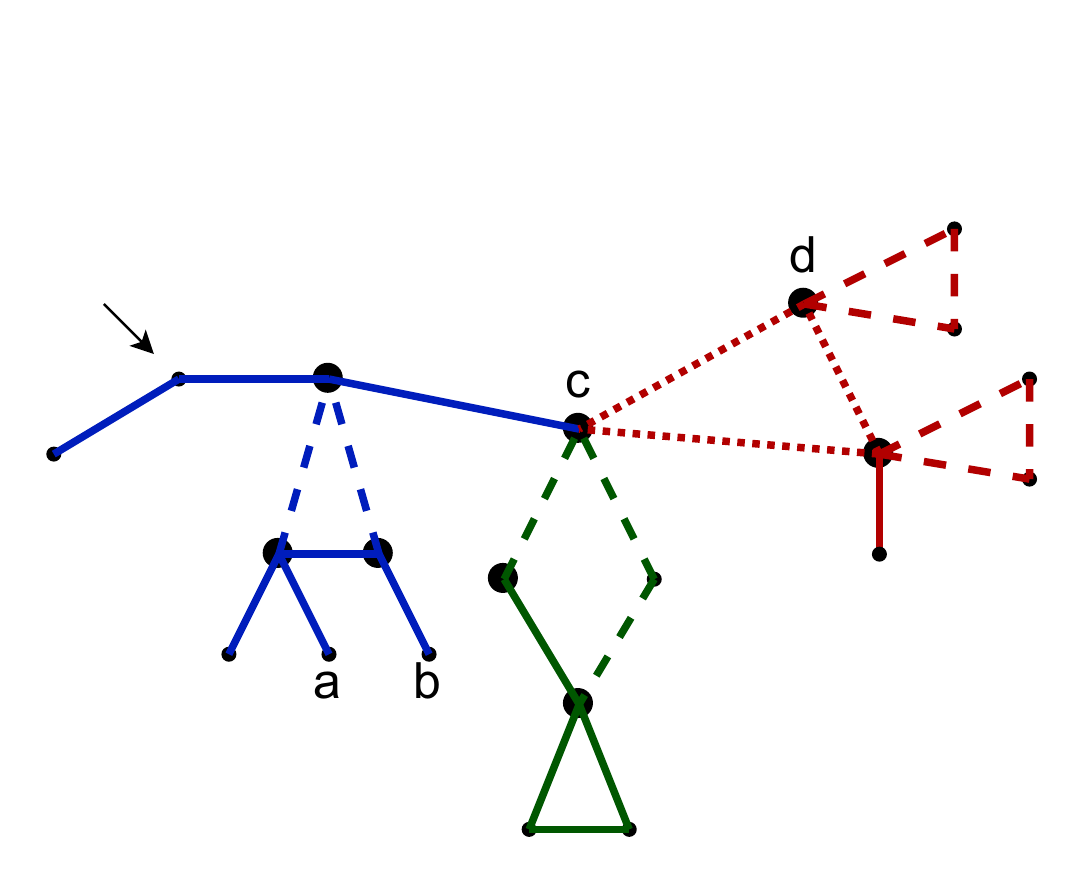}

\caption{
Two levels of a recursive decomposition satisfying Lemma \ref{lemma:decomposition}
.
Fragments from $D_j$ are marked with different colors and from $D_{j+1}$ with line styles.
The arrow indicates $r_G$, the root of~$G$.
Border vertices of $j$-th level are highlighted.
A buyer wishing to purchase an $a$-$b$ path is not assigned to $j$-th level, but a buyer interested in $a$-$c$ paths is.
Vertices $c$ and $d$ are connected at all levels of decomposition, so corresponding customers would be processed at its last level.
}
\label{fig:decomposition}
\end{wrapfigure}

The second invariant ensures that the number of levels is bounded by $\OO{ \log_k m}$.
By fixing $ k = \ceil{\log ^{\frac{1}{2}} m} $ we achieve an $\OO{ \frac{\log m}{ \log \log m}}$ bound on $L$.
We choose this value because some parts of the algorithm are exponential in $k$.

$D_{j+1}$ is a refinement of $D_j$ obtained by a two-phase procedure.
In the first step each fragment is split into subparts of balanced size.
The second phase refines this division in order to balance the number of border vertices in each resulting fragment.
Precise description of this process and the proof of Lemma \ref{lemma:decomposition} can be found in Appendix~\ref{subsect:decomp_apx}.


\subsection{Classification of buyers}

A pair of vertices $u$ and $v$ is said to be connected in a decomposition $D_j$ if there exists a path from $u$ to $v$ fully contained in a single fragment from $D_j$.
Buyer $i \in B$ will be processed at the last level where $u_i$ and $v_i$ are connected.
This way every buyer is assigned to a single level of decomposition.

\begin{remark}
	\label{rem:borderPass}
	If $j$ is the last level at which vertices $u$ and $v$ are connected, every $u$-$v$ path in the whole graph contains a border vertex on $j$-th level.
	\end{remark}


\section{Algorithm for a single decomposition}
\label{sect:algo}

In the previous section buyers have been divided into subsets by a recursive graph decomposition.
Now we focus on a single ($j$-th) level of decomposition.
By exploiting its properties our algorithm constructs prices which achieve a constant factor approximation of revenue with respect to buyers assigned to this level (denoted $B_j$).


The main idea behind the algorithm for a single decomposition is to split the paths desired by buyers into smaller sections and handle them separately.
In the following we define a partitioning of those paths and discuss that it suffices to be able to solve the natural two subcases. 

\subsection{The skeleton}
\label{subsect:skeleton}

\begin{definition}
\term{Skeleton on $j$-th level}, denoted $\mname{SK}_j$, is a minimal subgraph of $G$ containing all simple paths between border vertices of $j$-th level.
Equivalently, an edge belongs to the skeleton, i.e. is a \term{skeleton edge}, if and only if a simple path connecting two border vertices passes through it.
A vertex adjacent to a skeleton edge is a \term{skeleton vertex}.
\end{definition}

\begin{definition}
\label{def:nonSkeletonComponent}
A \term{non-skeleton component on $j$-th level} is a maximal connected subgraph of a fragment from $D_{j+1}$ containing no edges from $\mname{SK}_j$.
\end{definition}

Note that, by the definition of a border vertex, $\mname{SK}_{j+1}$ is always a superset of $\mname{SK}_j$ and $\mname{SK}_L = G$.
The following lemma allows for a clear distinction between the paths inside the skeleton and outside it.
The proof can be found in Appendix~\ref{sect:skeletonPath}.

\begin{lemma}
\label{lemma:skeletonPath}
Every simple path connecting two skeleton vertices passes only though skeleton edges. 
\end{lemma}

\begin{corollary}
\label{coll:skeletonRepresentative}
Each non-skeleton component contains exactly one skeleton vertex.
\end{corollary}

\begin{definition}
\label{def:skeletonRepresentative}
Let us define a \term{skeleton representative} of a vertex $v$ on $j$-th level denoted by $\repr_j(v)$.
If $v$ is a skeleton vertex in $D_j$, then $\repr_j(v) = v$.
Otherwise, the representative of $v$ is the unique skeleton vertex in the non-skeleton component on $j$-th level containing $v$.
\end{definition}

\begin{wrapfigure}[17]{r}{0.35\textwidth}
\vspace{-0.7cm}
\includegraphics[width=0.35\textwidth]{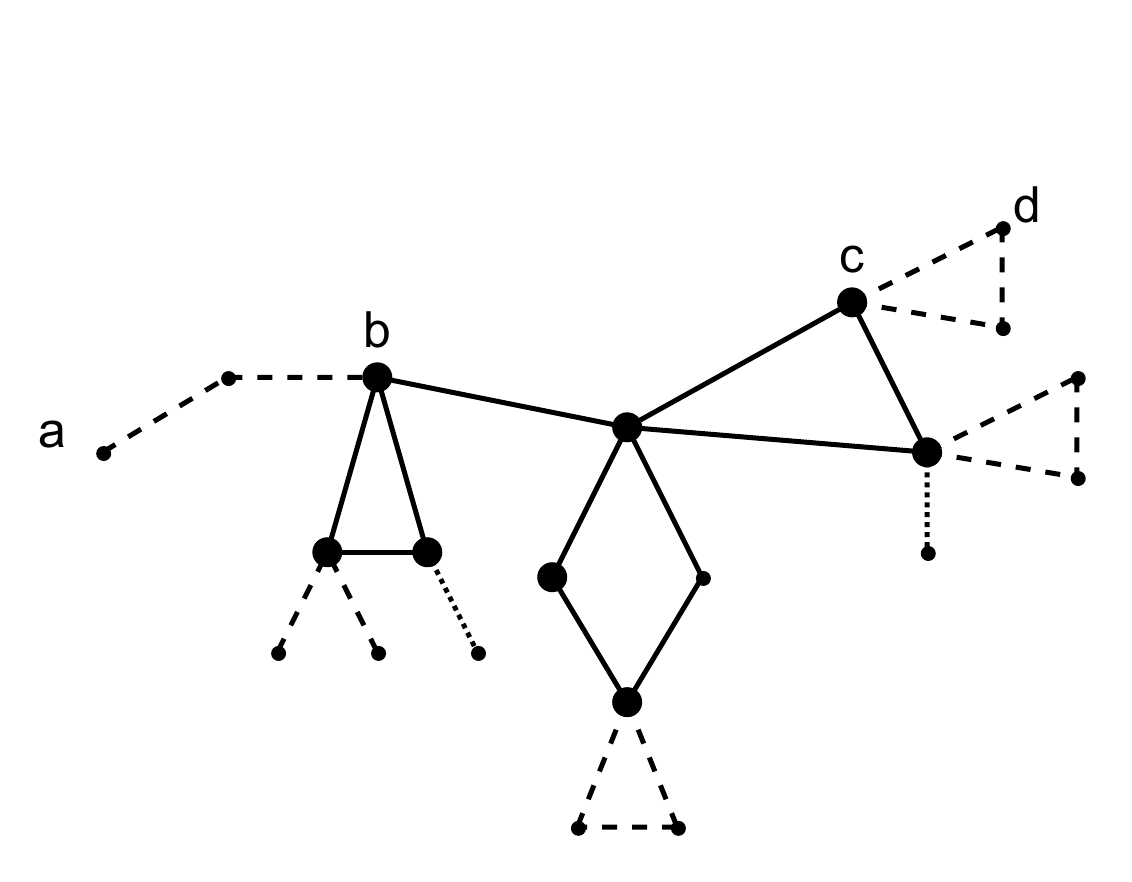}

\caption{The cactus from Figure \ref{fig:decomposition} with highlighted border vertices and the skeleton on $j$-th level.
Each connected group of edges dotted in a same style forms a single non-skeleton component.
Each $a$-$d$ path is split into a skeleton section (a $b$-$c$ path) and two non-skeleton sections: $a$-$b$ and $c$-$d$ paths.}

\end{wrapfigure}

Consider a buyer $i$ from $B_j$ wishing to buy the cheapest $u_i$-$v_i$ path.
Recall from Section \ref{sect:decomp} that each $u_i$-$v_i$ path contains at least one border vertex of $j$-th level.
By Corollary \ref{coll:skeletonRepresentative}, each path from $u_i$ to $v_i$ contains vertices $u_i$, $\repr_j(u_i)$, $\repr_j(v_i)$, $v_i$.
Although some of those four vertices may be equal, they are guaranteed to appear in this order.
This allows us to split every such path into three parts (some of which may be empty):
\begin{itemize}[nosep]
\item First \term{non-skeleton section} -- a simple path from $u_i$ to $\repr_j(u_i)$, which contains no skeleton edges.
\item A \term{skeleton section} -- a simple path from $\repr_j(u_i)$ to $\repr_j(v_i)$.
By Lemma \ref{lemma:skeletonPath}, it consists only of skeleton edges.
\item Second \term{non-skeleton section} -- a simple path from $\repr_j(v_i)$ to $v_i$.
Similarly to the first one, it does not contain skeleton edges.
\end{itemize}
\noindent Note that the endpoints of individual sections do not depend on the choice of the particular $u_i$-$v_i$ path.
Our algorithm uses this property to handle both kinds of sections individually.

\subsection{Splitting the graph into two independent subproblems}

The algorithm handles two subproblems: pricing the non-skeleton and skeleton edges to maximize revenue generated by respectively non-skeleton and skeleton sections.

\textbf{The skeleton subproblem:}
Consider a buyer $i \in B_j$ wishing to purchase the cheapest $u_i$-$v_i$ path.
In this subproblem she buys a cheapest $\repr_j(u_i)$-$\repr_j(v_i)$ path as long as its cost is at most $b_i$ (her original budget).
This situation achieved by setting the price of all non-skeleton edges to zero.

\textbf{The non-skeleton subproblem:}
In this case, we set the prices of all the skeleton edges to zero.
Each buyer $i \in B_j$ will purchase the cheapest paths from $u_i$ to $\repr_j(u_i)$ and from $v_i$ to $\repr_j(v_i)$ if their total cost does not exceed $b_i$. 

Let us introduce additional notation:
\begin{itemize}[nosep]
\item Let $\mname{OPT}_j$ be the maximal revenue obtained by any price vector and envy-free assignment of paths to the buyers from $B_j$.
\item Let $\mname{SKOPT}_j$ and $\mname{NSKOPT}_j$ be the maximal revenues for the skeleton and non-skeleton subproblem respectively.
\end{itemize}

Note that any envy-free solution for the whole graph immediately yields envy-free solutions for both subproblems.
Thus, $\mname{SKOPT}_j + \mname{NSKOPT}_j \geq \mname{OPT}_j$.
The algorithm solves both subproblems independently.
Then, the computed solutions are compared and the one with greater revenue is chosen.
Sections \ref{sect:skeleton} and \ref{sect:nonSkeleton} describe polynomial time approximation algorithms for the skeleton and non-skeleton subproblem respectively.

\section{Non-skeleton edges}
\label{sect:nonSkeleton}
\newcommand{\CrevToOpt}{4}

This section describes an algorithm for solving the non-skeleton subproblem on $j$-th level, that is pricing the non-skeleton edges and maximizing revenue generated by non-skeleton sections of paths allocated to buyers from $B_j$.
Prices of all edges in the skeleton on $j$-th level are set to zero.
On the last level of decomposition the skeleton contains the whole graph $G$.
Thus, we assume that $j < L$.
The algorithm presented here finds prices generating at least $\frac{\mname{NSKOPT}_j}{\CrevToOpt}$ revenue.

\subsection{The rooted case}
\label{subsect:rootedCase}

Before describing the method for pricing non-skeleton edges, let us discuss an easier problem, solution to which is a subprocedure used by the final algorithm.

\begin{definition}
\label{def:rootedInstance}
Consider an instance of the tollbooth problem on cactus graphs defined by a cactus $H$ and a set of buyers $B_H$.
We will say that it is a \term{rooted instance} if there exists a vertex in $H$, called root, which is an endpoint of every path desired by the buyers.
\end{definition}

\begin{definition}
Consider a buyer $i \in B_H$ in a rooted instance, who wishes to purchase a cheapest $u_i$-$v_i$ path.
Her \term{destination vertex} is the one of vertices $u_i$ and $v_i$ which is not the root.
\end{definition}

\begin{lemma}
\label{lemma:rootedCase}
Any rooted instance of the tollbooth problem on cactus graphs can be solved in polynomial time.
It is also true if the problem admits only those price assignments, under which the distances from the root to some vertices are equal to arbitrarily fixed constants. 
\end{lemma}

For every possible price assignment and envy-free allocation each buyer is assigned a shortest path from the root to her destination vertex as long as its cost does not exceed her budget.
Thus, presenting a polynomial algorithm for finding optimal prices is sufficient to prove the above lemma.

The algorithm is based on dynamic programming whose subproblems mimic the structure of the tree of biconnected components of $H$ rooted in $r$ -- the root from Definition \ref{def:rootedInstance}.
For each biconnected component $C$ it calculates values $dp_{C, d}$, which are defined as the maximum revenue generated by buyers whose destination vertices are in the subtree graph of $C$ (excluding its topmost vertex) under the assumption that the distance i.e., cost of a cheapest path, from $r$ to $C$ (its depth) equals $d$.
Note that the distance from $C$ to the root is in fact the distance between $r$ and the topmost vertex of $C$.
The following lemma allows us to consider only polynomially many values $d$.

\begin{lemma}
\label{lemma:depths}
For any rooted instance of the tollbooth problem on cactus graphs there exists an optimal solution, such that the distance from each vertex to the root belongs to the set $\mathcal{D}$ containing zero and buyers' budgets: 
$ \mathcal{D} = \set{0} \cup \longset{b_i}{i \in B_H}$.
\end{lemma}

In Section \ref{sect:skeleton} the algorithm needs to find optimal prices given constraints on the distance from certain vertices to the root.
The following corollary allows for handling such cases.
Both the lemma and corollary are proven in Appendix \ref{subsect:rootedCase_apx}.

\begin{corollary}
\label{coll:fixedDepths}
Consider a rooted instance of the tollbooth problem on cactus graphs and a subset $S$ of vertices of $H$ such that for each $v \in S$ required depth $e_v$ of this vertex is given.
Prices of edges in $H$ are said to be feasible if the cost of a cheapest $r$-$v$ path equals $e_v$ for each $v \in S$.
Let us assume that there is at least one such price assignment.
Then, there exists a feasible price assignment maximizing revenue for which the distance from $r$ to each vertex $v \not\in S$ belongs to the set $\mathcal{D'}$:
$$\mathcal{D'} = \set{0} \cup \longset{b_i}{i \in B_H} \cup \longset{e_i}{i \in S}$$
\end{corollary}

\subsubsection{Solution for a rooted instance}

The input to the procedure consists of a graph $H$, buyers $B_H$ and a (possibly empty) set of constraints $S$ from Corollary \ref{coll:fixedDepths}.
For each vertex $v$ we define a set of its possible depths $\mathcal{D}_v$ in the following way:
$$ \mathcal{D}_v = \begin{cases}
						\set{e_v}, & v \in S \\
						\set{0} \cup \longset{b_i}{i \in B_H} \cup \longset{e_u}{u \in S}, & v \not\in S \end{cases} $$

For each biconnected component $C$ the algorithm calculates the values of $dp_{C, d}$ for every $d \in \mathcal{D}_v$ where $v$ is the topmost vertex of $C$.
It is possible that some values of $d$ inevitably lead to violation of the constraints on depths of vertices from $S$.
In such a case we set $dp_{C, d} = - \infty$.
For simplicity, we also assume that $dp_{C, d} = - \infty$ for each $d \not\in \mathcal{D}_v$.

The biconnected components of $H$ are processed bottom up based on the structure of the tree of biconnected components.
The algorithm handles biconnected components differently depending on whether they consist of a single edge or a cycle.
The root component $R$, which is the root of the tree of biconnected components, contains the whole $H$ in its subtree graph and is treated in yet another way.
Let us introduce useful notation:
\begin{itemize}[nosep]
	\item $cnt_{v,x}$ -- the number of buyers whose budgets are at least $x$ and whose destination vertex is $v$.
	\item $\mathcal{C}_v$ -- the set of all biconnected components whose topmost vertex is $v$.
\end{itemize}

\textbf{The case of a single edge:}
Let us denote the lower vertex of the considered biconnected component $C$ as $v$ and the upper as $u$.
Note that the subtree graph of $C$ consists of the $(u,v)$ edge and subtree graphs of biconnected components from $\mathcal{C}_v$, which are edge disjoint and share only the topmost vertex.
All simple paths to the root from vertices contained in those subtree graphs pass though $v$.
Furthermore, each simple path from $v$ to the root must contain $u$.
Basing on those observations, the algorithm calculates $dp_{C, d}$ for each $d \in \mathcal{D}_u$ according to the following formula:
$$ dp_{C, d} = \max_{d' \in \mathcal{D}_v \; ;\; d' \geq d} \left( cnt_{d', v} \cdot d' + \sum_{C' \in \mathcal{C}_v} dp_{C', d'} \right)$$
The optimal value of $d'$ is stored along with $dp_{C, d}$ in order to find the prices after calculating optimal revenue ($(d' - d)$ is the price of the $(u,v)$ edge).

\textbf{The case of the root component:}
The root must have depth $0$, hence for this component only a single value ($dp_{R,0}$) is calculated:
$ dp_{R, 0} = \sum_{C' \in \mathcal{C}_r} dp_{C', 0}$.

\textbf{The case of a cycle:}
Let us denote the considered cycle by $C$, its topmost vertex as $v$ and its subtree graph by $G_C$.
$G_C$ consists of $C$ and the subtree graphs of components from $\mathcal{C}_u$ for all $u \in C \setminus \set{v}$.

Let us examine the structure of the subproblem.
All paths from vertices from $G_C$ to the root pass through $v$.
By construction of the tree of biconnected components each vertex $s$ in $G_C \setminus \set{v}$ belongs either to $C$ or to a subtree component of $C' \in \mathcal{C}_u$ for a unique $u \in C \setminus \set{v}$.
In the latter case every path from $s$ to the root also passes through $u$.
Edges which form possible $s$-$u$ paths belong to smaller subproblems, so given the depth of $u$, the optimal prices can be calculated.
Thus, now we are only interested in the depths of vertices in $C$ (more precisely, their distance to $v$, because it's depth is fixed).
Note that these distances depend only on the edges in $C$.

Consider any prices assigned to them.
Let $T$ be a shortest-path tree of $C$ rooted in $v$.
Exactly one edge from $C$ does not belong to $T$, we will say that it is unused.
After removing this edge, the cost of a cheapest path from any vertex in $G_C$ to $v$ does not change.

\begin{wrapfigure}{r}{0.4\textwidth}
		\vspace{-0.3cm}
		\includegraphics[width=0.4\textwidth]{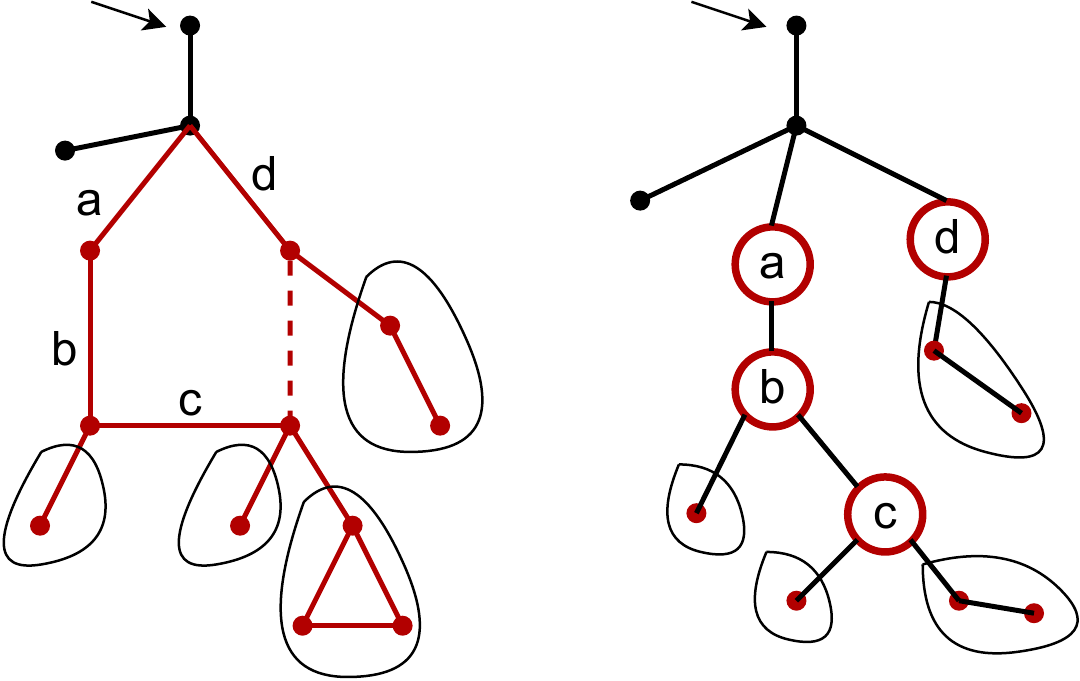}
		\caption{
		Processing a cycle for a fixed unused edge.
		Its subtree graph (without the topmost vertex) is highlighted in red in both the cactus graph (left) and the temporary tree of biconnected components (right).
		Solid lines enclose subproblems which already have been solved.
		Arrows indicate respectively the root of $G$ and the root component.
		The algorithm calculates the values of $dp^e$ for biconnected components $c$, $b$, $a$ and $d$ (in this order).
		}
		
	\end{wrapfigure}

The algorithm iterates over all edges in $C$ fixing the current one, denoted $e$, to be the unused edge.
In this step the algorithm finds an optimal solution among those price assignments which result in $e$ being unused.
First, the price of $e$ is set to $b_{max}+1$ ($b_{max} = \max \longset{b_i}{i \in B_H}$).
This effectively removes $e$ from the graph, as no buyer will ever purchase a path containing it.

Without $e$, $C$ becomes a path and edges in $C \setminus \set{e}$ constitute individual biconnected components.
The subtree graphs of components in $\mathcal{C}_u$ for $u \in C \setminus \set{v}$ remain intact, but now in the subtree graph they are descendants of the single-edge biconnected components from $C \setminus \set{e}$ instead of the cycle $C$.
However, the solutions calculated for them for every valid depth still remain valid.
Thus, the only biconnected components in $G_C \setminus \set{e}$ for which we need to calculate $dp_{C', d}$ are formed by single-edge biconnected components.
The algorithm calculates these values in a bottom up manner as described previously.
We assumed that $e$ is an unused edge, so let us denote the results as $dp^e$.

Let $C_1$ and $C_2$ be the biconnected components formed by edges of $C$ which are adjacent to $v$.
If $e$, the unused edge, happens to be adjacent to $v$, there is only one such component.
In this case $C_2$ is just a placeholder with an empty subtree graph and $dp^e_{C_2, d}$ equals zero for all $d \in \mathbb{R}$.
Note that the union of subtree graphs of $C_1$ and $C_2$ contains the same vertices and edges (except for $e$) as $G_C$.
Furthermore, the two subtree graphs can only share one vertex: $v$.
Thus, if we admit only such solutions, where $e$ is unused, then $dp_{C, d} = dp^e_{C_1, d} + dp^e_{C_2, d}$.
Since for every possible price assignment there exists an optimal allocation where one edge of $C$ is unused, it is enough to iterate over all possible edges $e \in C$:
$$ dp_{C, d} = \max_{e \in C} \left( dp^e_{C_1, d} + dp^e_{C_2, d} \right) $$
Using the above formula the algorithm computes $dp_{C, d}$ for every $d \in \mathcal{D}_v$.
Like previously, respective price assignments to edges of $C$ are stored alongside the results.


Since all the above procedures run in polynomial time and each biconnected component is processed only once, the solution is found in polynomial time.
Prices obtaining the computed maximal revenue can be easily calculated using additional information stored alongside the values of $dp_{C, d}$.  
This proves Lemma \ref{lemma:rootedCase}.

\subsection{The non-skeleton subproblem}

In order to solve the non-skeleton subproblem the algorithm utilizes a special structure of non-skeleton components on $j$-th level. 
For each of them, a rooted instance of the tollbooth problem on cactus graphs is created.
Those subproblems are solved by the procedure described in Section \ref{subsect:rootedCase}.
Resulting price assignments for individual non-skeleton components are merged by a probabilistic procedure which can, however, be derandomized.

\subsubsection{Constructing rooted instances}

Recall that each buyer $i \in B_j$ is defined by a triple $\left(u_i, v_i, b_i \right)$, which means that she has a valuation of $b_i$ for all $u_i$-$v_i$ paths.
Since the skeleton edges are given away for free, the algorithm only processes respective non-skeleton sections, which are modeled by two independent copies of $i$-th buyer: $\left( u_i, \repr_j(u_i), b_i \right)$ and $\left( v_i, \repr_j(v_i), b_i \right)$.
Each of them is added to the instance associated with the non-skeleton component containing $u_i$ and $v_i$ respectively.
If $u_i$ or $v_i$ is a skeleton vertex, the corresponding non-skeleton section is empty and can be ignored.  
Since $\repr_j(s)$ is the same for all vertices $s$ within a single non-skeleton component (Corollary \ref{coll:skeletonRepresentative}), all subproblems defined this way will indeed be rooted instances.

\begin{remark}
\label{rem:splittingBuyers}
It follows from the classification of buyers, that if $u_i$ and $v_i$ are not in $SK_j$, the non-skeleton components containing $u_i$ and $v_i$ belong to the same fragment of $D_j$, but to different fragments from $D_{j+1}$.
\end{remark}


\subsubsection{Algorithm for a single fragment}

The above observation allows us to treat all fragments in $D_{j}$ independently.
Let us consider a single fragment $H \in D_j$ and by $rev_{j, H}(p)$ let us denote the revenue generated by its selling non-skeleton edges for prices $p$ to buyers from $B_j$.
If for each buyer at most one non-skeleton section was non-empty, the rooted instance would be independent.
We could apply their solutions verbatim -- each buyer $i$ present in an instance would always be able to spend $b_i$ as assumed.
However, it's not the case -- for example a buyer $i$ may be present in two non-skeleton components and both solutions to the corresponding rooted instances may require her to pay $b_i$ for each non-skeleton section, in which case she would not buy anything.

We solve this issue by using a randomized procedure: each fragment $F \in D_{j+1}$ contained in $H$ is independently and equiprobably colored black or white.
Every non-skeleton component in a black fragment is priced according to the solution to the corresponding rooted instance.
All edges in non-skeleton components in white fragments are given away for free.

\begin{lemma}
\label{lemma:randomisedNonSkeleton}
Let $p$ be the price vector found by the above randomized algorithm and $q$ be any price vector feasible for the non-skeleton subproblem.
Then, the following inequality holds:
$$ \Exp{\rev_{j, H}(p)} \geq \frac{1}{\CrevToOpt} \rev_{j, H}(q) $$
\end{lemma}

This follows from the fact that for any non-skeleton section with probability at least $\frac{1}{4}$ it will be in the black fragment and the other non-skeleton section of the same buyer will be in a white one.
The deterministic algorithm iterates over all possible colorings and chooses the best one, which will yield results at least as good as the expected value.
Because there are at most $\OO{k}$ ($\OO{\sqrt{\log m}}$) fragments of the next level contained in $H$, this takes polynomial time.
Detailed proofs of the lemma and the corollary can be found in Appendix \ref{subsect:non_skeleton_subproblem_apx}.

\begin{corollary}
\label{lemma:derandomisedNonSkeleton}
There exists a deterministic polynomial algorithm which for a non-skeleton subproblem on $j$-th level finds prices achieving at least $\frac{\mname{NSKOPT}_j}{\CrevToOpt}$ revenue.
\end{corollary}

\let\CrevToOpt\undefined

\section{Skeleton edges}
\label{sect:skeleton}
\newcommand{\Crounding}{1024} 
\newcommand{\CrevToOpt}{2048} 
\newcommand{\CmyScoreToMaxScore}{32} 
\newcommand{\CrevToScore}{4} 
\newcommand{\CrevToMyScore}{128} 
\newcommand{\CroundedRevToMyScore}{512} 

This chapter describes an algorithm solving the skeleton subproblem on a single, $j$-th level of decomposition.
Non-skeleton edges do not influence envy-freeness of a solution because the are given away for free.
Thus, a buyer $i \in B_j$ wishing to buy an $u_i$-$v_i$ path can be thought of as a buyer with the same budget wishing to buy a shortest $\repr_j(u_i)$-$\repr_j(v_i)$ path.
This chapter describes a polynomial time algorithm which finds prices for edges in $\mname{SK}_j$ generating at least $\frac{\mname{SKOPT}_j}{\CrevToOpt}$ revenue.

\subsection{Decomposing the skeleton}
\label{subsect:decomposeSkeleton}

Let us begin by exploring important properties of the skeleton subproblem. 

\begin{remark}
\label{rem:locality}
For each buyer $i \in B_j$, there exists a fragment in $D_j$ containing the skeleton representatives of both $u_i$ and $v_i$.
Furthermore, every path between the skeleton representatives contains a border vertex of $j$-th level.
\end{remark}
This property is true for $u_i$ and $v_i$, which follows from the way buyers are assigned to levels of decomposition.
It also holds for their representatives because by Definition \ref{def:skeletonRepresentative} each vertex and its representative on $j$-th level belong to the same fragment of $D_{j+1}$.
In order to take advantage of this, we decompose $\mname{SK}_j$ into smaller subgraphs.

Consider a process of compressing $\mname{SK}_j$ by applying the following operations:
\begin{enumerate}[nosep]
\item
Let there be two edges: $(u,v)$, $(v,w)$ such that $v$ is not a border vertex and has degree equal two.
Merge them into a single edge $(u,w)$ and erase $v$ from the graph.
\item
For any two vertices $u$ and $v$, if there are parallel $(u,v)$ edges, merge them.
\end{enumerate}
The process concludes when none of the above operations can be executed anymore.

\begin{definition}
\label{def:segment}
A \term{segment on $j$-th level} is a subgraph of $\mname{SK}_j$ which is contracted into a single edge by the above procedure.
The \term{endpoints of a segment} are the endpoints of the corresponding edge in the compressed graph.
The \term{cost} or \term{length} of a segment is the cost of a cheapest path between its endpoints which is fully contained within the segment.
\end{definition}
Note that the only vertices shared by segments are their endpoints.
Furthermore, every border vertex on $j$-th level is an endpoint of a segment because border vertices are never erased from $\mname{SK}_j$ during the compression.


\begin{definition}
The \term{skeleton of a fragment} $F \in D_j$ is the minimal subgraph of $G$ containing all simple paths between skeleton vertices from $F$.
It is denoted $\mname{SK}_j(F)$.
\end{definition}
By Lemma \ref{lemma:skeletonPath} all simple paths between skeleton vertices belong to the skeleton, so $\mname{SK}_j(F) \subseteq \mname{SK}_j$.

\begin{definition}
\label{def:innerOuter}
Let $F$ be a fragment in $D_j$.
A segment which is contained in $F$ is said to be an \term{inner segment} of $\mname{SK}_j(F)$.
A simple path which connects two skeleton vertices of $F$ and contains no edges from $F$ is called an \term{outer extension} of $\mname{SK}_j(F)$.
\end{definition}
Each outer extension starts and ends in a border vertex.
Thus, it consists of several whole segments, i.e. traversed from one endpoint to another.
Also, together with the inner segments they form a partition of $\mname{SK}_j$, which we prove in Appendix~\ref{subsect:decomposeSkeleton_apx}.

\begin{SCfigure}

\centering
\includegraphics[width=0.5\textwidth]{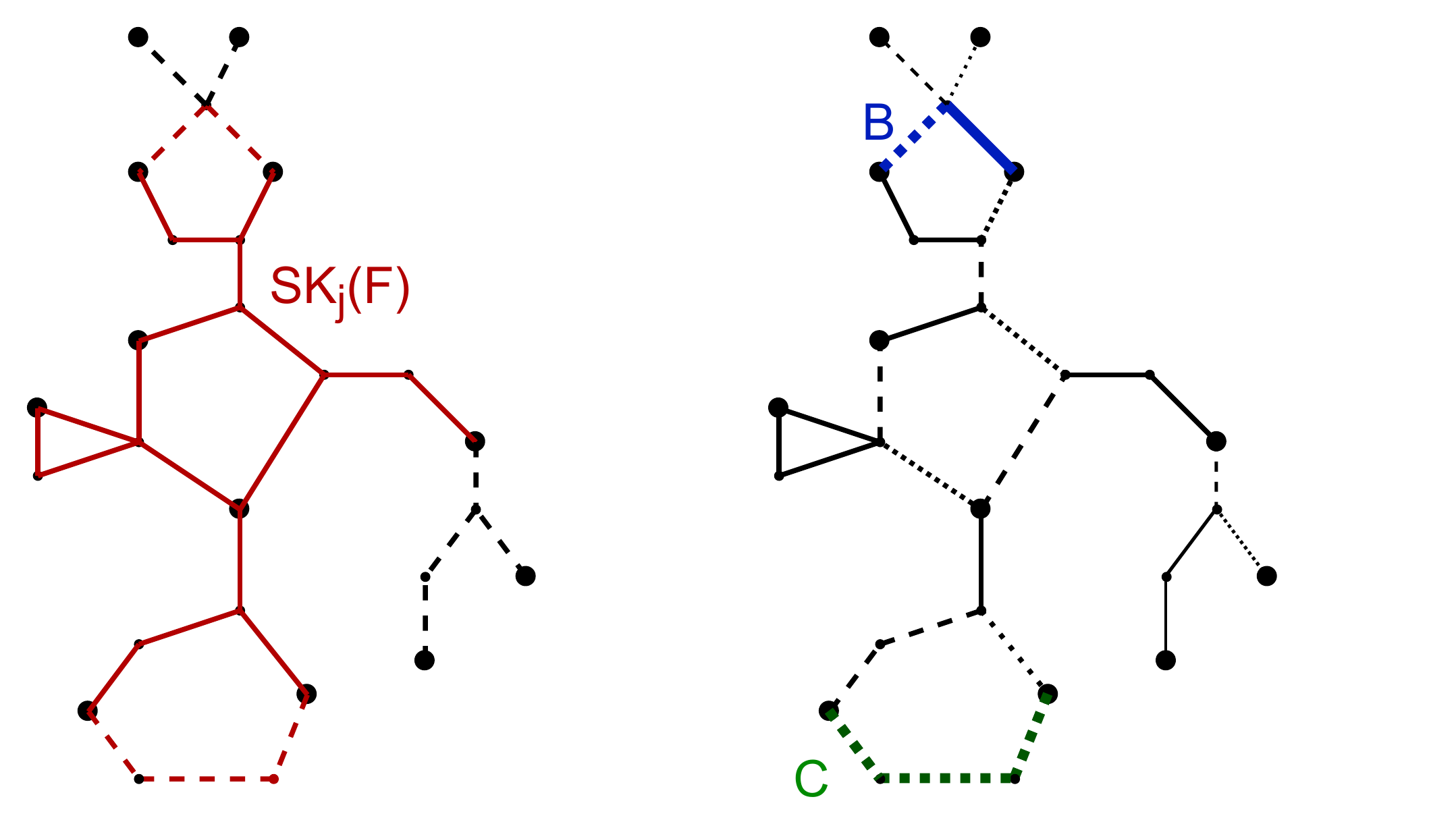}

\caption{An example skeleton on $j$-th level with highlighted border vertices.
On the left skeleton edges are grouped by the style of the lines according to fragments and on the right to segments on $j$-th level.
In the picture on the right edges from the fragment $F \in D_j$ are marked by solid lines.
$B$ and $C$ are the outer extensions of its skeleton, which is marked in red.
}

\end{SCfigure}

\begin{remark}
\label{rem:innerOuterPartition}
Inner  segments and outer extensions form an edge-disjoint partition of $\mname{SK}_j(F)$.
\end{remark}

\subsection{Price rounding}

If a simple path neither starts nor ends in a given segment, it will either traverse this segment from one endpoint to another or not at all.
Thus, for such paths it is sufficient to consider the cost of the segment and not the prices of individual edges.
Let us formalize this observation.

\begin{definition}
\label{def:involved}
A buyer $i \in B_j$ is said to be \term{involved} in a segment on $j$-th level if either $\repr_j(u_i)$ or $\repr_j(v_i)$ belongs to this segment and is not its endpoint.
\end{definition}  

\begin{remark}
For fixed costs of segments, changing prices of individual edges inside a segment influences only the involved buyers.
\end{remark}

Using this observation, the algorithm could guess the costs of segments and then distribute those costs among individual edges in a way that would maximize the revenue of involved buyers.
However, such a naive approach cannot be used because there are infinitely many combinations of segments' costs.
The following lemma allows for considering only finitely many of them.
Its full proof can be found in Appendix~\ref{subsect:priceRounding_apx}.
The main idea behind it is that optimal prices can be rounded down to fulfill the condition below without dramatically decreasing revenue. 

\begin{lemma}{\textbf{(rounding)}}
\label{lemma:rounding}
There exists a price assignment obtaining revenue of at least $\frac{\mname{SKOPT}_j}{4}$ such that each segment's length belongs to the following set:
$$ P = \longset{ \frac{m b_{max}}{2^t} }{ t \in \set{0, 1, \dots, \ceil{\log \left(\Crounding \cdot m^2 \cdot |B_j| \right)} } } \cup \set{0} $$
Here $b_{max}$ is the greatest budget of buyers in $B_j$ and $m$ is the number of edges in $G$.
\end{lemma}

\subsection{Pricing strategies}
\label{subsect:pricingStrategies}

Globally, the number of segments can potentially be large.
Thus, the algorithm can not explicitly iterate over all combinations of their rounded costs.
Instead, it uses the structure of paths desired by customers in $B_j$ to handle each fragment in the current decomposition separately.

Recall that for each buyer $i \in B_j$ there exists at least one fragment in $D_j$ such that both $\repr_j(u_i)$ and $\repr_j(v_i)$  belong to it (Remark \ref{rem:locality}).
The algorithm assigns each customer from $B_j$ to one of such fragments.
Let us fix a fragment $F \in D_j$ and denote the set of buyers assigned to it by $B_{j, F}$.
By definition, the skeleton of $F$ contains the skeleton sections of all paths desired by buyers from $B_{j,F}$.
Thus, revenue generated by those buyers in the skeleton subproblem under any prices $p$, denoted $\rev_{j,F}(p)$, depends only on the costs of edges in $\mname{SK}_j(F)$.

Having said that, particular fragments in $D_j$ cannot be processed completely independently, as their skeletons overlap.
More specifically, each outer extension of $\mname{SK}_j(F)$ consist of several inner segments from different fragments in $D_j$.
However, simple paths starting and ending in $F$ traverse the outer extensions only as a whole, from one endpoint to another.
Thus, the costs of individual edges, or even segments, forming the outer extensions of $\mname{SK}_j(F)$ do not influence $\rev_{j,F}(p)$.
Only the total costs of outer extensions themselves do.
An outer extension's length is a sum of costs of several inner segments from different fragments.
Thus, even if considering only price assignments satisfying the rounding lemma (\ref{lemma:rounding}), it may not belong to $P$.
Because of this, they need to be treated differently than inner segments.

\begin{definition}
For a fragment $F \in D_j$ a \term{pricing strategy} $s$ for its skeleton edges is defined by:
\begin{itemize}[nosep]
\item A single number $p_{s,i} \in P$ for each inner segment $i$ of $\mname{SK}_j(F)$.
\item Two consecutive numbers $l_{s,o}, r_{s,o}$ from $P'$ for each outer extension $o$ of $\mname{SK}_j(F)$.
$$ P' = \longset{\frac{m^2 b_{max}}{2^t} }{ t \in \set{0, 1, \dots, \ceil{\log \left(\Crounding \cdot m^3 \cdot |B_j| \right)} } } \cup \set{0} $$
Alternatively, $l_{s,o}$ and $r_{s,o}$ can be both equal to zero.
\end{itemize}
Prices of skeleton edges in $F$ \term{implement} a strategy $s$ if the length of each inner segment $i$ in $\mname{SK}_j$ is $p_{s,i}$.
A global price assignment \term{implements} $s$ if additionally the length of each outer extension $o$ of $\mname{SK}_j(F)$ is in the interval $\left( l_{s, o}, r_{s, o} \right]$ or equals zero if $l_{s, o} =  r_{s, o} = 0$.
A combination of pricing strategies is \term{valid} if the costs assigned to inner segments satisfy the requirements on the lengths of outer extensions.
\end{definition}

On the high level, the algorithm for solving the skeleton subproblem works in two phases.
First, for each fragment and pricing strategy near-optimal prices implementing that strategy are found.
Then, the algorithm constructs a valid combination of strategies generating high overall revenue.

\subsubsection{Approximating revenue}

For each $F \in D_j$, the algorithm iterates over all strategies and finds near-optimal prices of edges in inner segments of $\mname{SK}_j(F)$ implementing them.
It is done without assuming exact lengths of outer extensions in $\mname{SK}_j(F)$, but only intervals of their possible values.
Thus, it is impossible calculate the revenue generated by $B_{j,F}$.
One can approximate it, though.

\begin{definition}
\label{def:approxRevenue}
Consider a fragment $F \in D_j$, pricing strategy $s$ for it and any prices implementing $s$, denoted $p$.
The \term{approximate revenue} generated by those prices, denoted $\rev_{j,F,s}(p)$, is the revenue generated by customers from $B_{j,F}$ under the following assumptions:
\begin{itemize}[nosep]
\item The skeleton edges of $F$ are priced according to $p$.
\item Each outer extension $o$ in $\mname{SK}_j(F)$ has length $r_{s,o}$.
\end{itemize}
\end{definition}

The following lemma allows the algorithm to focus on maximizing the approximate revenue calculated locally for each fragment and pricing strategy.
Its proof can be found in Appendix \ref{subsubsect:approx_rev_apx}.

\begin{lemma}
\label{lemma:outerOverpricing}
Consider a valid combination of pricing strategies and denote the strategy for a fragment $F \in D_j$ as $s_F$.
Then, for any price assignment $p$ implementing that combination of strategies:
$$ \sum_{F \in D_j} \rev_{j, F, s_F}(p) \geq \frac{b_{max}}{512}
\quad \implies \quad
\rev(p) \geq \frac{1}{4} \sum_{F \in D_j} \rev_{j, F, s_F}(p)$$
\end{lemma}

\subsubsection{Bounding the number of pricing strategies}

Since for each fragment the algorithm iterates over all pricing strategies, we need to argue that there is only polynomially many of them.
This follows from the fact that there are only $\OO{k}$ inner  segments and outer extensions in the skeleton of any fragment.
Proofs of the following claims can be found in Appendix \ref{subsubsect:strategyBounding_apx}.

\begin{lemma}
\label{lemma:partSegmentation}
For each fragment $F \in D_j$, there are $\OO{k}$ inner  segments and outer extensions in its skeleton.
\end{lemma}

\begin{corollary}
\label{coll:strategiesBound}
The number of pricing strategies for a fragment in $D_j$ is polynomial in $m$ and $n$. 
\end{corollary}

\subsection{Solution for a single fragment and a fixed pricing strategy}
\label{subsect:singleFragment}

In this section we present a polynomial procedure which for a fixed fragment $F \in D_j$ and a pricing strategy $s_F$ finds prices of skeleton edges in $F$ implementing $s_F$.
Prices $p$ found using this method achieve approximate revenue ($rev_{j, F, s_F}(p)$) of at least $\frac{1}{\CmyScoreToMaxScore}$ of the maximal approximate revenue possible for $s_F$.
For ease of presentation we first describe a randomized procedure and then derandomize it.

The algorithm handles each inner segment in $SK_j(F)$ separately employing one of two procedures depending on whether the segment contains cycles or not.
Let us introduce necessary notation.
$S$~is the segment in question, $l$ and $r$ its endpoints and $c$ its length determined by $s_F$.
Since $c$ is fixed, prices of individual edges in $S$ are only relevant for those buyers from $B_{j,F}$ who are involved in $S$ (denoted $B_S$).
If other buyers' paths pass through $S$, they will pay $c$ (the price for a shortest $l$-$r$ path).
For each buyer $i \in B_{j,F}$ we assume (without loss of generality) that $u_i$ is the endpoint lying inside $S$. 
As far as approximate revenue is concerned, fixing a pricing strategy for $F$ determines the distances between all segments' endpoints in $F$.
Using the corresponding metric (denoted $dist_{s_F}$) for each $i \in B_{j,F}$ we define $b_{i,l}$ as $b_i-\min \set{dist_{s_F}(l, v'_i), dist_{s_F}(l, v''_i)}$ where $v'_i$ and $v''_i$ are the endpoints of the other segment buyer $i$ is involved in.
If no such segment exists, $v'_i = v''_i = \repr_j(v_i)$.
Note that $b_{i,l}$ is an upper bound on how much buyer $i$ can spend on edges in $S$ if her chosen path passes though $l$.
In a similar way for $r$ we define $b_{i,r}$.

\subsubsection{Cyclic segments}
\label{subsubsect:cyclicSegments}
	
Since $G$ is a cactus, there is no path between $l$ and $r$ outside $S$.
Thus, $B_S$ can be split into two disjoint sets: $B_{S, l}$ and $B_{S, r}$ of buyers $i$ whose all $\repr_j(u_i)$-$\repr_j(v_i)$ paths contain respectively $l$ or $r$.
Prices for edges in $S$ are chosen equiprobably from the following four solutions:

\begin{wrapfigure}[11]{r}{0.4\textwidth}
	\vspace{-0.5cm}
	\includegraphics[clip=true,width=0.4\textwidth]{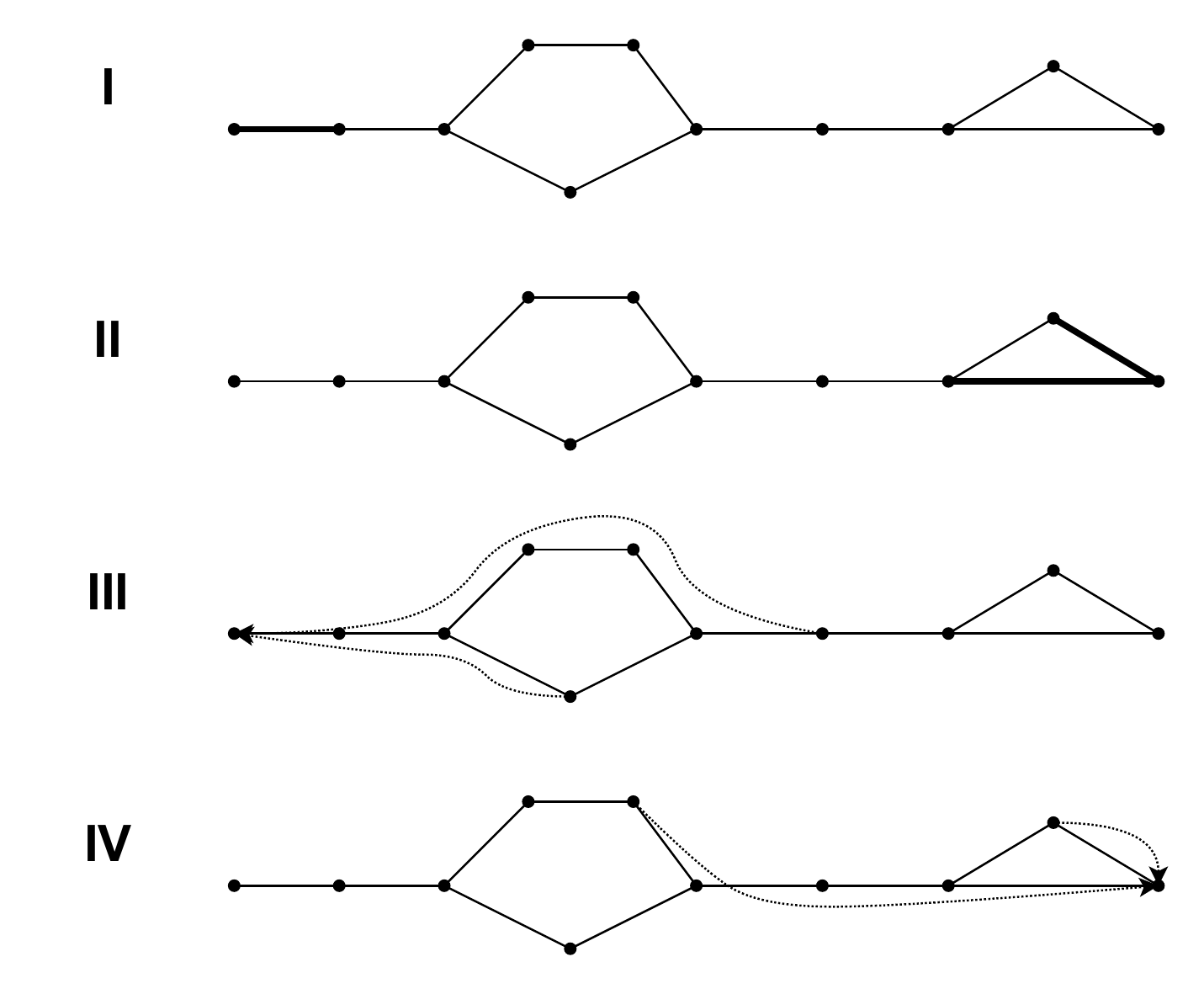}
	\end{wrapfigure}

\textbf{\nth{1}:}
Set the price of all edges in $S$ adjacent to $l$ to $c$ and of the others to zero.
This way the buyers from $B_{S, l}$ will be forced to pay $c$ for edges in $S$, but the ones from $B_{S, r}$ will not have to pay anything.

\textbf{\nth{2}:}
Set the price of all edges in $S$ adjacent to $r$ to $c$ and of the others to zero.

\textbf{\nth{3}:}
Create a rooted instance from $S$ where each buyer $i \in B_{S, l}$ is represented as $\left(l, u_i, b_{i,l} \right)$ and the depth of $r$ is set to $c$.
Then, apply its optimal solution, which by Corollary \ref{coll:fixedDepths} can be found in polynomial time.
As $b_{i,l}$ is an upper bound on how much buyer $i$ can spend in $S$, optimum revenue of the rooted instance is an upper bound on the contribution of selling edges in $S$ to $B_{S, l}$ towards approximate revenue.
Thanks to the first two solutions in the other segment she is involved in, with probability at least $\frac{1}{4}$ she will be able to spend exactly $b_{i,l}$ .
Thus, in expectation, buyers from $B_{S, l}$ spend in $S$ at most four times less than in an optimal solution implementing $s_F$.

\textbf{\nth{4}:}
Symmetrically to the previous solution, for the buyers from $B_{S, r}$.

\subsubsection{Acyclic segments}
\label{subsubsect:acyclicSegments}

By construction of the segments, if $S$ has no cycles, it must be a path.
However, a buyer from $B_S$ can potentially desire multiple paths, some passing though $l$ and some though $r$.
The algorithm chooses one of the following solutions each with probability $\frac{1}{4}$.

\begin{wrapfigure}{r}{0.4\textwidth}
\vspace{-0.5cm}
	\includegraphics[clip = true,width=0.4\textwidth]{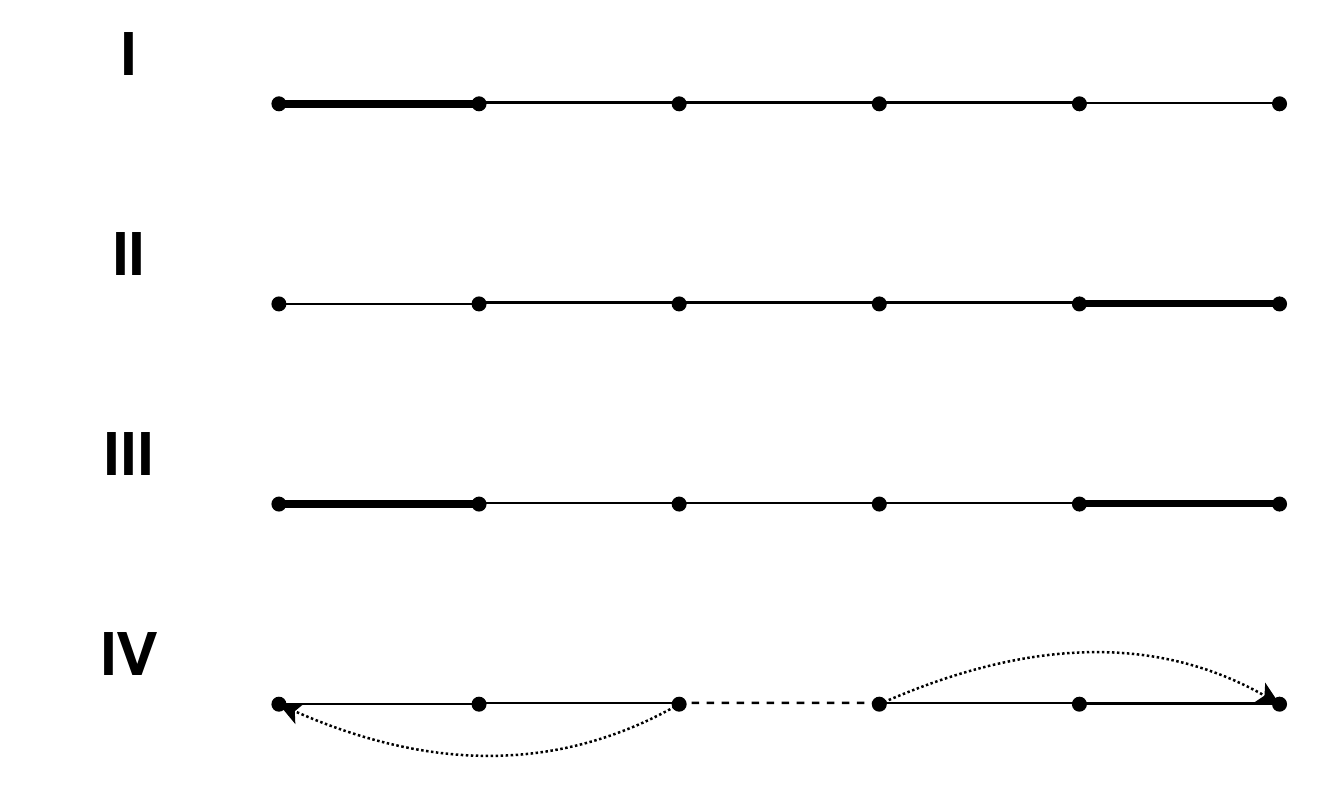}
	\end{wrapfigure}

\textbf{\nth{1} and \nth{2}:}
As in the case of a cyclic segment, the edge adjacent to $l$ or $r$ is assigned cost $c$ and the rest zero.

\textbf{\nth{3}:}
The prices of the left- and rightmost edges are set to $\frac{c}{2}$.
If $S$ has only one edge, its cost is set to $c$.
Each buyer $i \in B_S$ whose $b_{i,l}$ or $b_{i,r}$ is at least $\frac{c}{2}$ will pay for edges in $S$ at least half of what she pays in any solution (she cannot pay more than $c$).
That is if she does not pay anything in the other segment she's involved in, which happens with probability at least $\frac{1}{4}$.
Otherwise, she may refrain from buying any path at all, because it may be too expensive.

\textbf{\nth{4}:}
We restrict our attention to buyers in $i \in B_S$ for which $\max \set{b_{i,l}, b_{i,r}} < \frac{c}{2}$ (denoted $B'_S$).
For any prices of edges in $S$ implementing $s_F$, sets of vertices reachable inside $S$ from $l$ and $r$ with cost smaller than $\frac{c}{2}$ must be disjoint.
This splits $B'_S$ into disjoint sets $B'_{S, l}$ and $B'_{S, r}$ of buyers $i$ whose vertices $u_i$ belong to respectively the former or the latter set.
Note that edges sold to $B'_{S,l}$ and $B'_{S, r}$ are disjoint and separated by at least one edge not sold to anyone in $B'_S$ -- a \textbf{pivot edge}.

Our algorithm iterates over each edge $e$ in $S$ and finds a near-optimal price assignment among these with $e$ being a pivot edge.
First, $S$ is split into two parts: to the left of $e$ (with $l$) and to the right of it (with $r$).
Each buyer $i \in B'_S$ is assigned to $B'_{S, l}$ or $B'_{S, r}$ based on the part $u_i$ belongs to.
Then, a rooted instance for the left part with buyers $\left(l, u_i, b_{i,l} \right)$ for each $i \in B'_{S, l}$ is created.
For the right part, the algorithm creates an analogous instance.
Their solutions constitute a valid solution for $S$ -- buyers' budgets are smaller than $\frac{c}{2}$ and so are lengths of both parts of $S$.
Their endpoints are also endpoints of $e$, so this makes $e$ a pivot edge.
It's price will be positive to keep the cost of $S$ equal $c$.

Both rooted instances assume maximal possible budgets of buyers from $B'_S$ for edges in $S$ (for the fixed pivot edge).
Thus, the total revenue from them is not lower than what buyers in $B'_S$ could ever pay for edges in $S$ with $e$ being a pivot edge.
Like previously, each buyer is able to pay the assumed amount of money with probability at least $\frac{1}{4}$, so in expectation a quarter of this revenue is obtained by the randomized solution.
Since each price assignment results in at least one pivot edge, taking maximum over all possible pivot edges results in approximate revenue from $B'_S$ in $S$ being in expectation at most four times smaller than in an optimal solution implementing $s_F$.


\begin{lemma}
	\label{lemma:singleFragment}
	Let $p$ and $q$ be two price assignments implementing $s_F$ such that $q$ maximizes $\rev_{j, F, s_F}$ and $p$ is the result of the randomized algorithm for pricing the skeleton edges of a single fragment.
	Then, the following inequality holds:
	$$ \Exp{ \rev_{j, F, s_F}(p) } \geq \frac{1}{\CmyScoreToMaxScore} \rev_{j, F, s_F}(q)$$
	\end{lemma}

We prove this lemma by splitting the revenue generated by buyers into two categories: revenue generated by selling whole segments from one endpoint to another and by selling parts of segments.
In a given segment, the latter is generated only by buyers involved in it.
The third and fourth options guarantee that those buyers in expectation contribute to the approximate revenue at least $\frac{1}{32}$ of their contribution under prices $q$.
The first two guarantee that with probability at least $\frac{1}{16}$ a~path optimal for a buyer under prices $q$ is not more expensive under prices $p$.
This happens when she does not have to pay anything in the segments she is involved in.
Thus, in that case she would pay for whole segments as much as she would under prices $q$.
Hence, for selling the whole segments prices $p$ generate in expectation at $\frac{1}{16}$ of what the optimal solution does.
A more detailed proof of this lemma can be found in Appendix~\ref{subsect:singleFragment_apx}.

By Lemma \ref{lemma:partSegmentation} there are $\OO{\sqrt{\log m}}$ segments in $F$, so the algorithm can iterate over all $4^{\OO{\sqrt{\log m}}}$ possible solutions of the randomized procedure and pick the best one.

\begin{corollary}
	\label{coll:detSingleFragment}
	For $F \in D_j$, let $q$ be a price vector implementing $s_F$ which maximizes $\rev_{j, F, s_F}$.
	There exists a deterministic polynomial procedure which finds prices $p$ implementing $s_F$ such that:
	$$\rev_{j, F, s_F}(p) \geq \frac{1}{\CmyScoreToMaxScore} \rev_{j, F, s_F}(q)$$
	\end{corollary}

\subsection{Constructing a global price assignment}
\label{subsect:wholeGraph}

After applying the procedure from the previous section to each fragment $F \in D_j$ and every pricing strategy, the algorithm merges the resulting partial solutions into a global price assignment for~$\mname{SK}_j$.
The procedure presented in this section finds a valid combination of strategies maximizing total approximate revenue (score) obtained by respective price assignments in individual fragments.

\begin{definition}
\label{def:score}
Let $F$ be a fragment in $D_j$, $s$ a pricing strategy for $F$ and $p$ the prices constructed by the procedure from Section \ref{subsect:singleFragment} for $F$ and $s$.
Then, the \term{score} of $s$ equals $\rev_{j, F, s}(p)$.
The \term{score} of a valid combination of pricing strategies for a set of fragments is defined as a sum of respective scores.
\end{definition}

The only assumptions a pricing strategy for a given fragment $F \in D_j$ makes on the costs of edges outside $F$ is by setting intervals for lengths of outer extensions of $\mname{SK}_j(F)$.
Thus, to solve dependencies between fragments, it is enough to consider the outer extensions.
By definition an outer extension $o$ in $\mname{SK}_j(F)$ is a path outside $F$ between two border vertices of $F$, denoted $u$ and $v$.
Since $F$ is connected, it contains another $u$-$v$ path.
Thus, $u$ and $v$ lie on a simple cycle formed by $o$ and some of the inner segments of $\mname{SK}_j(F)$.
By giving bounds on the length of $o$ and setting the costs of inner segments, a pricing strategy for $F$ in fact imposes constraints on its length.
Thus, by controlling lengths of cycles split between multiple fragments we can solve the dependencies between fragments.

\begin{wrapfigure}{r}{0.4\textwidth}
\includegraphics[width=0.4\textwidth]{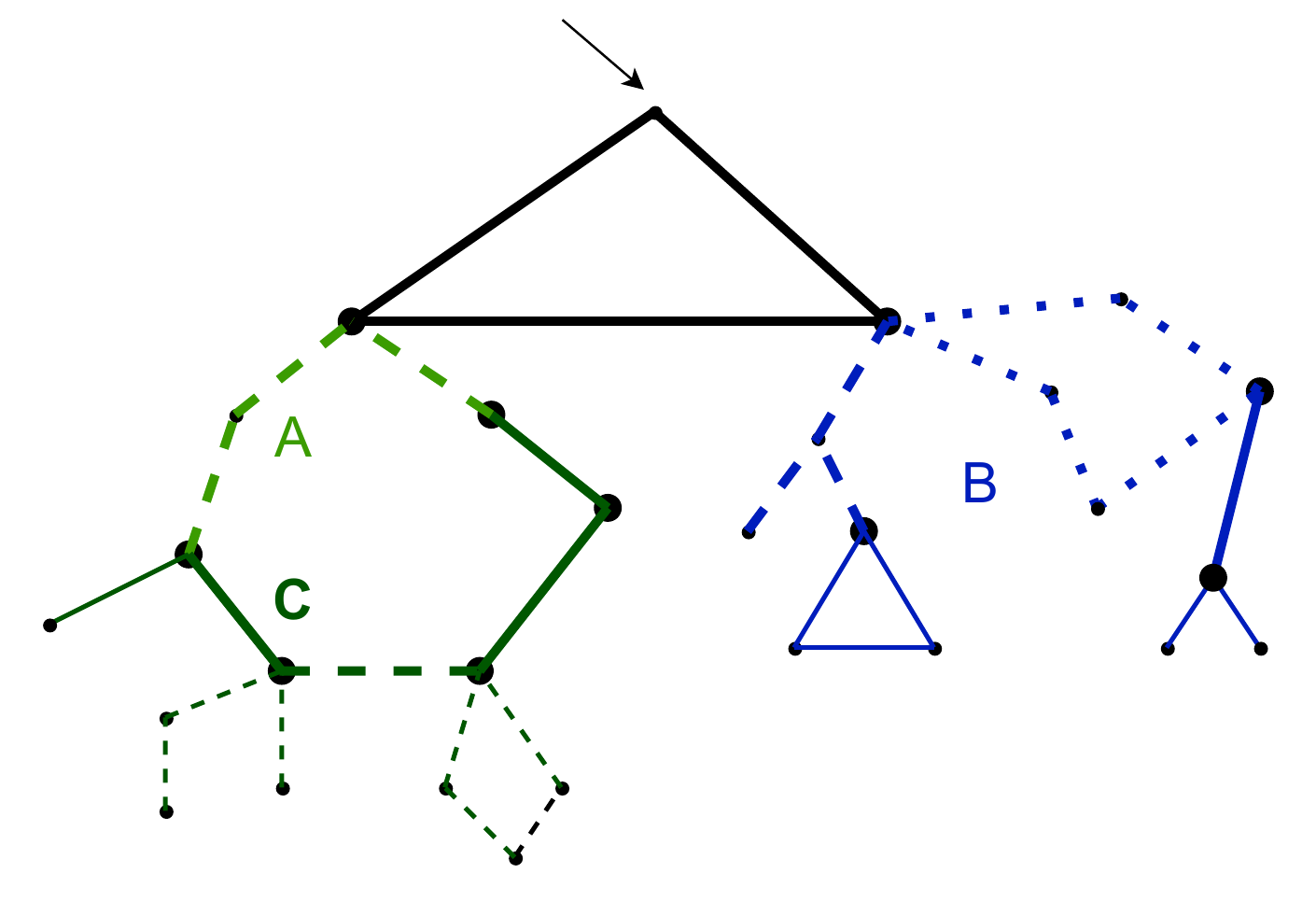}
\caption{
The whole graph forms a closed subproblem, whose topmost vertex is the root of $G$ indicated by the arrow.
It has two closed child subproblems -- $A$ (dark and light green) and $B$ (blue).
The former has exactly one child subproblem -- $C$ (light green), which is an open subproblem.
}
\end{wrapfigure}

A score-maximizing valid combination of pricing strategies is found by a recursive procedure, whose subproblems are sets of fragments in $D_j$.
For each such $S \subseteq D_j$ by $G_S$ we denote the graph formed by fragments in $S$.
It will fall into exactly one of the following categories:
\begin{itemize}[nosep]
	\item $G_S$ is an edge-disjoint union of subtree graphs of several biconnected components which share a topmost vertex.
	Then, $S$ is a \term{closed subproblem}.
	\item For a certain cycle $C$ that is split into multiple fragments, $G_S$ consists of all edges in $C$ except those contained in the same fragment as its topmost edges,
	and of the subtree graphs with topmost vertices in $C$, whose topmost edges don't belong to that fragment.
	Then, $S$ is an \term{open subproblem}.
\end{itemize}
For a closed subproblem $S$, $G_S$ and $G \setminus G_S$ do not share any cycles, so combinations of strategies for $S$ and $D_j \setminus S$ are independent.
Thus, a solution for $S$ is a single score-maximizing combination of strategies.
For an open subproblem the algorithm calculates multiple score-maximizing combinations for different assumptions about the costs of segments forming $C$.

\subsubsection{Open subproblems}
\label{subsubsect:open}

Let $u$ and $v$ be the two vertices shared by $G_S$ and $G \setminus G_S$.
There are two simple $u$-$v$ paths: the \term{upper path} in $G \setminus G_S$ and the \term{lower path} contained in $G_S$.
Since $C$ is the only cycle shared between the two subgraphs, the upper path is the only part of $G \setminus G_S$ belonging to outer extensions of $\mname{SK}_j(F)$ for $F \in S$.
The same holds for the lower path, $G_S$ and $D_j \setminus S$.
Thus, the lengths of those paths determine the compatibility of strategy combinations for $S$ and $D_j \setminus S$.

\begin{lemma}
\label{lemma:polynomialLengths}
Let $\mathcal{L}$ be the set of possible lengths of any simple path between border vertices under prices satisfying the rounding lemma (\ref{lemma:rounding}).
There are only polynomially many elements of $\mathcal{L}$.
\end{lemma}

The algorithm iterates over all possible pairs of lengths of the upper and lower path.
For each of them it constructs a score-maximizing combination of strategies.
By Lemma \ref{lemma:polynomialLengths}, the number of possibilities ($\abs{\mathcal{L}}^2$) is polynomial.
Its proof is in Appendix \ref{subsect:wholeGraph_apx}.

Let us denote current lengths of the upper and lower path as respectively $up$ and $low$.
Let the \term{closing fragment} be the fragment in $D_j$ containing the upper path and the \term{lower fragments} the ones forming the lower path.
We denote them as $F_1, F_2, \dots F_q$ in the order of appearance on the lower path from $u$ to $v$.
The remaining fragments do not intersect with $C$ and will be processed by recursive calls.
Corresponding subproblems are formed by the connected components of $G_S$ remaining after removing the lower fragments.
Let $H$ be such a component.
Naturally, it's a union of some fragments from $S_H \subseteq S$.
If it doesn't share any cycles with the lower fragments, $S_H$ must form a closed subproblem.
Otherwise, let $C'$ be a cycle shared between $H$ and the lower fragments.
Note that at least one of the topmost edges of $C'$ must belong to upper fragments.
By construction of the decomposition, both topmost edges must belong to the same lower fragment, denoted $F_a$.
Thus, exactly one lower fragment is present in $C'$ and no other cycle is shared between $H$ and the lower fragments.
This makes $S_H$ an open subproblem and $F_a$ its closing fragment.

\textbf{Closed child subproblems:}
For each such subproblem our procedure returns a valid score-maximizing combination of strategies, which is applied verbatim to fragments forming it.

\textbf{Open child subproblems:}
For each open subproblem our procedure returns up to $\abs{\mathcal{L}}^2$ strategy combinations -- each one maximizing the score for certain lengths of its upper and lower path.
For a fixed strategy in $F_a$ the former is already determined and latter must be within a certain interval.
Since no other fragment from $D_j \setminus S_H$is present in $C'$, those are the only requirements on the solution to $S_H$.
Thus, for each possible pricing strategy in $F_a$, the algorithm finds the best compatible solution for $S_H$ and adds its score to the score of the pricing strategy.
This way the algorithm extends pricing strategies for $F_a$ to cover $H$.  
  
\textbf{Lower fragments:}
We have reduced the problem of solving $S$ to finding a valid score-maximizing combination of strategies for the lower fragments.
Each of them has exactly one outer extension which is still relevant -- the one forming $C$.
All the others are lower paths of open child subproblems and have already been handled.
Thus, it is enough to find a combination of strategies which makes the lower path's length equal $low$ and is compatible with the upper path costing $up$.
Hence, for each lower fragment we consider only pricing strategies compatible with the length of $C$ being equal $low + up$.
It remains to choose a score-maximizing combination of those strategies such that the costs of segments forming the lower path sum up to $low$, which is done using dynamic programming inspired by the knapsack problem.

Let $u_0, u_1, \dots u_q$ be such vertices that $u_0 = u$, $u_q = v$ and for the other $i$'s $u_i$ is shared by $F_i$ and $F_{i+1}$.
By $score_{up+low, i, l}$ we denote the score of the strategy for $F_i$, under which the inner segments forming the $u_{i-1}$-$u_i$ path have total cost $l$ and the interval for the outer extension consisting of the rest of $C$ contains $up + low - l$.
If there is no such strategy, we set $score_{up + low, i,l} = - \infty$.
For each $l \in \mathcal{L}$ and $i \leq q$, let $dp_{up + low, i, l}$ be the maximum score of a combination of strategies for segments $F_1, F_2, \dots F_i$, which sets the cost of the $u_0$-$u_i$ path contained in the lower path to $l$. 
Our algorithm calculates the values of $dp$ using the following formula:
$$ dp_{up + low, i, l} = \max_{d \in \mathcal{L}} score_{up + low, i,d} + dp_{up + low, i - 1, l - d}$$
Note that $dp_{up + low, q, low}$ is the maximum score of any valid combination of pricing strategies for fragments in $S$ with the lower path of length $low$ and the upper of length $up$.
Along with the values of $dp$, our algorithm stores an optimal combination of strategies.

\subsubsection{Closed subproblems}
\label{subsubsect:closed}

Let $v$ be the topmost vertex of $G_S$.
Since the fragments in $S$ containing $v$ do not share any other vertex, their skeletons do not overlap and their pricing strategies may be chosen independently.
Let us consider the connected components of $G_S$ resulting from removing those fragments.
Similarly as in the case of an open subproblem, each such component forms either an open or closed child subproblem, which is solved by recursive calls to respective procedures.
Solutions to the closed subproblems are applied verbatim.
For each open child subproblem one of the fragments containing $v$ is the closing fragment.
Thus, they are processed in the same way as previously: with each pricing strategy for a fragment containing $v$ we associate the score-maximizing compatible solution to each subproblem, whose closing fragment it is.

For each fragment in $S$ containing $v$ the algorithm iterates over all possible strategies for it and chooses one with the greatest score (including the scores from respective child subproblems).
This determines a valid score-maximizing combination of pricing strategies for all fragments in $S$.

The whole graph is a closed subproblem, a solution to which defines the global solution.

\subsection{Summing up}

The above algorithm constructed a price assignment for all skeleton edges.
In this section we show that its revenue is indeed only a constant factor away from $\mname{SKOPT}_j$.

First, we have restricted our attention to prices, for which costs of inner segments on $j$-th level belong to the set $P$ from Lemma~\ref{lemma:rounding}.
Let $Q_j$ denote the set of all such price assignments.
Then, for every $F \in D_j$ and every possible strategy $s_F$ the algorithm found prices $p_{s_F}$ achieving approximate revenue at most $32$ times smaller than the maximal approximate revenue under $s_F$ (Corollary \ref{coll:detSingleFragment}).
In the end our algorithm constructs a global solution $p$ which corresponds to a score-maximizing combination of strategies.
Thus, the total approximate revenue of $p$ is at most $32$ times smaller then the maximal approximate revenue.
In the following $s_{q,F}$ denotes the strategy implemented by prices $q$ in $F$ -- note each $q \in Q_j$ implements exactly one combination of strategies.

\begin{equation}
\label{math:scoreLowerBound}
\sum_{F \in D_j} \rev_{j, F, s_{p,F}}(p) \geq \frac{1}{\CmyScoreToMaxScore} \max_{q \in Q_j} \sum_{F \in D_j} \rev_{j, F, s_{q,F}}(q)
\end{equation}

By Lemma \ref{lemma:outerOverpricing} high score (approximate revenue) of the combination of pricing strategies results in high revenue.
The following lemma provides a lower bound for approximate revenue.
A full proof can be found in Appendix~\ref{subsect:SummingUp_apx}.
The observation behind it is that the main difference between the approximate revenue and the real one is overestimating the price of outer extensions by a multiplicative factor of $2$ or a relatively small additive factor.
Thus, by decreasing the optimal prices by those factors, we can guarantee this overestimation not to make the paths too expensive for buyers who generate revenue in the optimal solution.

\begin{lemma}
\label{lemma:maxScoreLowerBound}
There exists a price assignment $q \in Q_j$ with total approximate revenue at most $\CrevToScore$ times smaller than the maximal revenue obtained by prices in $Q_j$.
\end{lemma}

By combining Inequality \ref{math:scoreLowerBound}, Lemma \ref{lemma:maxScoreLowerBound} and the rounding lemma (\ref{lemma:rounding}), we achieve the following bound on the score of $p$.
$$ \sum_{F \in D_j} \rev_{j, F, s_{p,F}}(p) \geq 
\frac{1}{\CmyScoreToMaxScore} \max_{q \in Q_j} \sum_{F \in D_j} \rev_{j, F, s_{q,F}}(q) \geq
\frac{1}{\CrevToScore \cdot \CmyScoreToMaxScore} \max_{q' \in Q} \rev_j(q') \geq
\frac{1}{4 \cdot 32 \cdot 4} \mname{SKOPT}_j$$

Obviously, $\mname{SKOPT}_j \geq b_{max}$, so by Lemma \ref{lemma:outerOverpricing}:
$$ \rev_j(p) \geq \frac{1}{4} \sum_{F \in D_j} \rev_{j, F, s_{p,F}}(p) \geq \frac{1}{\CrevToOpt} \mname{SKOPT}_j $$

By proving this inequality we have shown that the polynomial algorithm for the skeleton subproblem achieves constant approximation ratio.

\let\Crounding\undefined
\let\CrevToMyScore\undefined
\let\CrevToOpt\undefined
\let\CmyScoreToMaxScore\undefined
\let\CrevToScore\undefined
\let\CroundedRevToMyScore\undefined

\section{Concluding remarks}
\label{sect:conclusion}
In Sections \ref{sect:nonSkeleton} and \ref{sect:skeleton} we have presented polynomial time constant factor approximation algorithms for the non-skeleton and skeleton subproblems.
Thus, we have shown that prices achieving at least a constant fraction of optimal revenue can be found in polynomial time for each of the $L$ levels of decomposition.
Recall that by setting $k = \ceil{\log^{\frac{1}{2}}{m}}$ we ensure $L$ to be $\OO{\frac{\log m}{\log \log m}}$.
Hence, our polynomial algorithm for the tollbooth problem on cactus graphs yields an $\OO{\frac{\log m}{\log \log m}}$ approximation guarantee on revenue.

It remains an open question whether there exist polynomial time algorithms giving sublogarithmic guarantees on revenue for further generalizations of the tollbooth problem, for example for the cases where the underlying graphs are only assumed to have bounded treewidth.

\bibliography{bibliography}

\newpage
\appendix

\section{Omitted proofs and technical details}
\label{sect:details}
\subsection{Balanced decomposition}
\label{subsect:decomp_apx}

Here we present in full detail the procedure for constructing the balanced decomposition, which completes the proof of the following lemma.

\newtheorem*{lemma:decomposition}{Lemma \ref{lemma:decomposition}}

\begin{lemma:decomposition}
Consider a family of decompositions $D_1, D_2, \dots D_L$ of a cactus graph $G$ satisfying the following invariants for each valid $j$:
\begin{enumerate}[nosep]
\item Each fragment in $D_j$ is split into $\OO{k}$ fragments in $D_{j+1}$.
\item The maximal number of edges in a fragment forming $D_{j+1}$ is $\OMG{k}$ times smaller than in $D_j$.
\item Each fragment forming $D_j$ contains at most $\OO{k}$ border vertices of $j$-th level.
\item Each pair of associated edges belongs to the same fragment of $D_j$.
\item All fragments forming $D_j$ are connected subgraphs of $G$.
\end{enumerate}
For $k$ being an unbounded and nondecreasing function of $m$ (the number of edges in $G$), such a family can be found in polynomial time.
\end{lemma:decomposition}
	
	Decompositions from the above lemma are defined recursively with $D_1$ consisting of a single part -- the whole graph $G$.
	The procedure for obtaining $D_{j+1}$ from $D_j$ has two phases, the first one reduces the number of edges and the second the number of border vertices.
	
	\subsubsection{Reducing the number of edges}
	\label{subsubsec:reduceEdges}
	We devise a procedure for splitting each fragment in such a way that the size of subparts in bounded.
	As all fragments are treated in the same way, we describe an algorithm which for one of them (denoted $F = \left<V, E \right>$) finds a decomposition satisfying the following conditions: 
	\begin{itemize}
	\item The number of subparts is between $\flfrac{k}{4}$ and $k$.
	\item The number of edges in each subpart is at most $4 \cdot \cefrac{|E|}{k}$.
	\item Associated edges belong to the same subpart.
	\end{itemize}
	
	The main idea is to partition $F$ into connected subgraphs with the number of edges between $\cefrac{|E|}{k}$ and $ 4 \cdot \cefrac{|E|}{k}$.
	We allow a single exception, one of the subparts may have fewer edges.
	From this the bound on the number of subparts follows immediately.
	
	All cycles in $F$ contain a pair of associated edges.
	By erasing one of them from each cycle, we can obtain a tree.
	Let us call this tree $T$.
	The algorithm first greedily partitions the edges of $T$ and then adds edges from $F \setminus T$ into respective subparts, so that the fourth invariant from Lemma \ref{lemma:decomposition} is preserved.
	
	The algorithm for partitioning the tree also needs to deal with pairs of associated edges, which may still be present in $T$.
	This is because the cycles of $G$, on which the associated edges lie, do not have to be fully contained in $F$ so it is possible that a pair of associated edges is not a part of any cycle in $F$.

	Let us root $T$ in one of its topmost vertices (there can be two of them).
	This way every pair of associated edges present in $T$ shares their upper endpoint.
	Using this property the algorithm makes sure they end up in the same subpart.
	
	The algorithm for partitioning $T$ maintains two kinds of subparts: open, to which edges of $T$ can still be added and closed, to which they cannot.
	This approach yields the following procedure for partitioning edges from $F$:

	\begin{enumerate}[noitemsep]
		\item From each pair of associated edges which is lying on a cycle in $F$ erase an arbitrarily chosen edge. Let us denote the resulting tree by $T$.
		\item Root $T$ in one of its topmost vertices and process every vertex $v$ in the order of non-increasing depths:
		\begin{enumerate}[nosep]
			\item For every child $u$ of $v$ in $T$: If there is an open subpart containing $u$, add the $(u,v)$ edge to it. Otherwise, initialize a new open subpart consisting only of $(u,v)$.
			\item Merge all pairs of subparts from the previous step containing edges associated with each other.
			\item Mark each open subpart as closed if it contains at least $\cefrac{|E|}{k}$ edges.
			\item As long as there are at least two open subparts, merge them. If their size reached $\cefrac{|E|}{k}$, mark the resulting subpart as closed. 
		\end{enumerate}
		\item Add each edge from $F \setminus T$ to the subpart containing the edge associated with it.
	\end{enumerate}

	After $v$ has been processed in the second step there is at most one open subpart in the subtree of $v$ and it contains strictly less than $\cefrac{|E|}{k}$ edges.
	Hence, open subparts created in the step 2a contain at most $\cefrac{|E|}{k}$ of them.
	This results in the size of closed subparts created in the second step being between $\cefrac{|E|}{k}$ and $2 \cefrac{|E|}{k}$.
	When edges from $F \setminus T$ are added, the subparts created so far can at most double their size because each edge is associated with at most one other edge.
	Thus, each subpart can contain at most $4 \cefrac{|E|}{k}$ edges.
	Furthermore, after the algorithm concludes, there may be only one subpart with fewer than $\cefrac{|E|}{k}$ edges.

	\subsubsection{Reducing the number of border vertices}
	\label{subsubsec:reduceBorder}
	
	Each subpart of $G$ created in the previous step already contains a number of border vertices, those include border vertices of $(j-1)$-th level as well as those resulting from splitting fragments from $D_j$ in the previous step (Section \ref{subsubsec:reduceEdges}). 
	Here we refer to them as old border vertices.
	Together with the new ones resulting from splitting subparts from Section \ref{subsubsec:reduceEdges}, they are border vertices of $j$-th level.
	We restrict our attention to a single subpart $F'$ generated by the previous step and assume that it contains at most $24k$ old border vertices.
	Let us define a procedure for splitting $F'$ into at most three connected subparts in such a way that:
	\begin{itemize}
		\item Each subpart contains at most $18k$ border vertices of $j$-th level.
		\item Edges associated with each other are in the same subpart.
	\end{itemize}
	
	Let $b$ be the number of old border vertices in $F'$.
	First, $F'$ is rooted in one of its vertices and a local tree of its biconnected components is created.
	Although the tree is constructed locally (as if $F'$ was the whole graph $G$), pairs of associated edges are still defined globally (Definition \ref{def:associated} still refers to the global root of $G$).
	Then, among the biconnected components whose subtree graphs in the local tree of biconnected components contain more than $\flfrac{b}{2}$ old border vertices, the algorithm chooses the one with a minimal subtree graph.
	Since the subtree of the root component is the whole fragment with $b$ border vertices, the subtree graph of at least one biconnected component must satisfy this condition.
	We fix one such component and call it the pivot component.
	If it is not the root component, it can be either a single edge or a cycle.
	If it is, we handle it as in the case of an edge, where the root vertex is considered to be the lower vertex.
	
	\textbf{The case of a single edge:}
	Consider the connected components of $F'$ resulting from removing the lower vertex $v$ of the pivot component.
	Each such subgraph contains at most $\cefrac{b}{2} + 1$ old border vertices (we include $v$ in all of them).
	This partition can, however, split pairs of associated edges.
	Let edges $(v,u_1)$ and $(v, u_2)$ be associated and belong to distinct subparts.
	Note that in this case both subparts can only have one edge adjacent to $v$ -- otherwise these edges would lie on at least two cycles in $G$.
	If the number of old border vertices contained in the union of those two subparts is at most $\cefrac{b}{2} + 1$, the two subparts are merged.
	Otherwise, edge $(v, u_1)$ is transferred to the subpart of $(v, u_2)$.
	In this case the subpart of $u_2$ gained only one additional vertex and the subpart of $u_1$ is still connected, as the edge $(v, u_1)$ must have been its only edge adjacent to $v$.
	
	In the partition defined this way every subpart (fragment on $(j+1)$-th level of decomposition) contains at most $\cefrac{b}{2} + 2$ old border vertices.
	Furthermore, at most two new border vertices are created in $F'$ on $j$-th level (one is $v$ and the other one results from transferring the associated edge to another subpart, which can happen only once).
	Hence, there are at most $\cefrac{b}{2} + 4$ border vertices on level $j$ in each subpart.
	
	However, such a partition may result in too many fragments on $(j+1)$-th level contained in $F'$.
	Thus, as long as there are two subparts whose union contains no more than $\cefrac{b}{2} + 4$ border vertices of $j$-th level, they are merged.
	Afterwords, there can be at most three subparts in the partition. 
	
	\textbf{The case of a cycle:}
	Let us denote the cycle being the pivot component as $C$.
	Consider connected components of $F'$ resulting from removing all the edges of $C$.
	
	If no such component has more than $\flfrac{b}{2}$ old border vertices, we can choose a number of consecutive vertices on $C$ such that their connected components contain between $\flfrac{b}{3}$ and $\cefrac{2b}{3}$ old border vertices in total.
	Let us paint all the vertices in their connected components black and the rest of the vertices in $F$ white.
	All the edges between pairs of two black and two white vertices are painted respectively black and white.
	Note that both the black and white subgraphs are connected.
	All white edges will form the first subpart and black ones the second.
	The edges between vertices of different colors can be assigned to either one.
	It is possible not to split pairs of associated edges since they share a vertex and thus cannot be both painted and have distinct colors.
	
	Each subpart contains at most $\cefrac{2b}{3}$ old border vertices.
	New border vertices in $F'$ can be created only by the edges connecting vertices with distinct colors.
	From the construction it follows that there are exactly two such edges and hence at most two new border vertices in $F'$.

	If one of the connected components resulting from removing edges of $C$ contains at least $\flfrac{b}{2} + 1$ old border vertices, we implement a different approach.
	Let $v$ be the vertex from $C$ belonging to this component.
	First, we add all edges from $C$ back to the graph, then remove $v$ from it.
	Note that each resulting connected component contains at most $\cefrac{b}{2} + 1$ old border vertices (including $v$ if it's an old border vertex).
	Thus, we can use the previously described procedure for the case of a single edge, where $v$ is treated like the lower vertex.
	
	\textbf{}
	Regardless of whether the pivot component is a cycle or a single edge, the above procedure splits $F'$ into at most three subparts.
	Each of them contains at most $\cefrac{2b}{3} + 2$ border vertices of $j$-th level.
	With $b \leq 24 k$ this gives an upper bound of $18k$ on their number in a single subpart of $F'$.

	\subsubsection{Recursive decomposition}
	
	Using subprocedures from Sections \ref{subsubsec:reduceEdges} and \ref{subsubsec:reduceBorder} we can define the recursive decomposition.
	Each fragment of $D_j$ is split into subparts in two phases.
	The first one produces subparts which have a limited number of edges.
	The second splits each of those subparts into at most three fragments to limit the number of border vertices on $j$-th level in each fragment forming $D_{j+1}$.
	Those fragments form $D_{j+1}$ and the process continues until all fragments consist of at most two edges.
	
	Recall that both phases create connected subparts and do not split associated edges.
	The first step produces at most $k$ subparts and the second splits each of those subparts into at most three smaller ones.
	Thus, each fragment in $D_j$ is split into $3k$ fragments in $D_{j+1}$.
	Reducing the number of edges during the first phase guarantees that the second invariant holds too.
	The following lemmas imply that our construction satisfies the third invariant and thus conclude the proof of Lemma \ref{lemma:decomposition}.

	\begin{lemma}
	\label{lemma:vertexShare}
	Consider a fragment $F$ split into $s$ edge-disjoint connected subparts in such a way, that the topmost edges of each biconnected component belong to the same subpart.
	Then, at most $2s-2$ vertices belong to more than one subpart.   
	\end{lemma}
	
	\begin{proof}
	
	Consider a cycle $C$ in $F$, whose edges belong to multiple subparts.
	Since the topmost edges of $C$ must belong to the same subpart, at least one of them does not contain any topmost edge.
	Let us denote it by $S$.
	In such a case, the topmost vertex of $C$ does not belong to $S$, so $S$ must be fully contained in the subtree graph of $C$.
	Thus, for fixed $S$ there can be only one such cycle $C$.
	Furthermore, if $S$ contains a topmost vertex of $F$, no such cycle exists.
	Hence, there are at most $s-1$ cycles whose edges belong to multiple subparts.
	
	Let us transform $F$ into a tree $T$ by removing one of the topmost edges from each cycle.
	Consider a cycle $C$ and let $S$ be the subpart containing its topmost edges.
	If all edges in $C$ belong to $S$, this subpart is still connected.
	Otherwise, let $v$ be the topmost vertex and $u$ be the other endpoint of the erased edge.
	If $u$ is incident to another edge of $S$, $S$ no longer is connected.
	Let us denote the two connected components of $S$ as $S_1$ and $S_2$.
	This can happen in only $q \leq s - 1$ cycles.
	Otherwise, the remaining edges of $S$ still form a connected subgraph.
	
	In the tree $T$ there are $s + q$ connected subparts, so there are $s + q - 1$ vertices shared by them.
	Let us check how this changes when a previously removed $(u,v)$ edge is added to the cycle $C$.
	If all edges in $C$ belong to the same subpart ($S$), no vertex is added to $S$.
	Let us consider the opposite case.
	If $u$ has been incident to another edge in $S$, the subparts $S_1$ and $S_2$ are in fact a single subpart $S$, which contained $u$ and $v$ even before adding the $(u,v)$ edge.
	Otherwise, $u$ is added to $S$ and shared with the subparts it previously belonged to.
	This can happen at most $s - 1 - q$ times.
	
	Summarizing, at most $2s - 2$ vertices of $F$ belong to multiple subparts contained in $F$.
	\end{proof}
	
	\begin{lemma}
	\label{lemma:borderVertexBound}
	For all $j$ each fragment of $D_j$ contains at most $26k$ border vertices of level $j$.
	\end{lemma}

	\begin{proof}
	
		In order to prove the lemma we will show that each fragment of $D_j$ contains at most $18 k$ border vertices of level $j-1$.
		This is sufficient because each fragment of $D_j$ is split to at most $3k$ fragments of $D_{j+1}$, which by Lemma \ref{lemma:vertexShare} can only increase the number of border vertices to $24k - 2$.
		We will prove this by induction.

		For $j = 1$ the statement is trivially true because there are no border vertices of level $0$ ($D_1$ consists of only a single fragment -- the whole graph).
		
		Consider any fragment $F \in D_{j+1}$ for $j > 0$ by $F'$ we denote the fragment in $D_j$ being a superset of $F$.
		Let us assume that $F'$ contains at most $18k$ border vertices of level $j-1$.
		The procedure for reducing the number of edges has split it into at most $k$ subparts.
		Thus, by Lemma \ref{lemma:vertexShare} the one containing $F$ cannot have more than $20k - 2$ old border vertices.
		Hence, the assumptions of the procedure for reducing the number of border vertices are satisfied and it produces fragments of $D_{j+1}$ (including $F$) containing at most $18k$ border vertices on level $j$.

	\end{proof}

\subsection{The skeleton}
\label{sect:skeletonPath}

Here we provide a proof of Lemma \ref{lemma:skeletonPath} from Section \ref{subsect:skeleton}.

\newtheorem*{lemma:skeletonPath}{Lemma \ref{lemma:skeletonPath}}
\begin{lemma:skeletonPath}
	Every simple path connecting two skeleton vertices passes only though skeleton edges. 
	\end{lemma:skeletonPath}

\begin{proof}
	We assume that there are at least two border vertices of $j$-th level.
	Otherwise, the lemma is trivially true.
	The proof is by contradiction.
	Let us say there is a simple path between two skeleton vertices passing though a number of non-skeleton edges.
	Consider a maximal part of this path consisting only of non-skeleton edges.
	It must start and end at skeleton vertices, let us denote them by $u$ and $v$.
	From the construction of $\mname{SK}_j$ it follows that there exists another $u$-$v$ path fully contained in the skeleton.
	Those two paths are of course disjoint, so $u$ and $v$ lie on a simple cycle.
	Consider an arbitrary skeleton edge on this simple cycle and a path between two border vertices, $s$ and $t$, which passes through it.
	Let $u'$ be such vertex on this $s$-$t$ path that lies on the cycle and is the closest one to $u$.
	Similarly, let $v'$ be the one closest to $v$.
	Without loss of generality we assume that $u'$ appears before $v'$ on the $s$-$t$ path.
	Consider the following path:
	\begin{itemize}
	\item From $s$ to $u'$ along the original $s$-$t$ path.
	\item From $u'$ to $u$ along the skeleton $u$-$v$ path.
	\item From $u$ to $v$ along the non-skeleton $u$-$v$ path.
	\item From $v$ to $v'$ along the skeleton $u$-$v$ path. 
	\item From $v'$ to $t$ along the original $s$-$t$ path.
	\end{itemize}
	It follows from the choice of $u'$ and $v'$ that the above path has no cycles (is simple).
	This leads to contradiction and concludes the proof.
	\end{proof}

\newcommand{\CrevToOpt}{4}
\subsection{The rooted case}
\label{subsect:rootedCase_apx}

Here follow the proofs omitted in Section \ref{subsect:rootedCase}

\newtheorem*{lemma:depths}{Lemma \ref{lemma:depths}}
\begin{lemma:depths}
For any rooted instance of the tollbooth problem on cactus graphs there exists an optimal solution, such that the distance from each vertex to the root belongs to the set $\mathcal{D}$ containing zero and buyers' budgets.
$$ \mathcal{D} = \set{0} \cup \longset{b_i}{i \in B_H}$$
\end{lemma:depths}

\begin{proof}
Let us call a price assignment regular if it satisfies the condition from the above lemma.
Consider optimal prices that are not regular.
We prove the lemma by showing the existence of a regular price assignment which generates at least as much revenue.

By $T$ let us denote a shortest-path tree of $H$ rooted in $r$, which is an acyclic subgraph of $H$ containing a shortest path from every vertex to the root.
Let us increase the price of all edges in $H \setminus T$ to the maximal budget $b_{max} = \max B_H$.
While it does not change the distances from vertices to the root, it guarantees that each vertex has a shortest path to the root contained in $T$ under any regular prices.

Let us consider a vertex $v$ and let $d_v$ be the distance from $v$ to $r$ under initial prices.
We set $d'_v$ to the smallest element of $\mathcal{D}$ greater or equal than $d_v$ or to $b_{max}$ if no such element exists.
Note that if for any two vertices $u$ and $v$ $d_u \geq d_v$, then $d'_u \geq d'_v$.
Hence, the values of $d'_v$ are non-decreasing on the paths from $r$ to every leaf in $T$.
Thus, there exist prices for edges in $T$ such that the distance from $r$ to $v$ equals $d'_v$ for every $v$.

Consider a vertex $v$ and customers whose destination vertex is $v$.
If $d'_v > d_v$, none of them has a budget from $[d_v, d'_v)$.
Hence, they now generate no less revenue than under the original prices.
If $d'_v < d_v$, then $d_v > b_{max}$, so in this case no buyer could have been allocated a $v$-$r$ path under the original prices.
Thus, the regular prices result in at least as much revenue as optimal ones.
\end{proof}


\newtheorem*{coll:fixedDepths}{Corollary \ref{coll:fixedDepths}}
\begin{coll:fixedDepths}
Consider a rooted instance of the tollbooth problem on cactus graphs and a subset $S$ of vertices of $H$ such that for each $v \in S$ required depth $e_v$ of this vertex is given.
Prices of edges in $H$ are said to be feasible if the cost of a cheapest $r$-$v$ path equals $e_v$ for each $v \in S$.
Let us assume, that there is at least one such price assignment.
Then, there exists a feasible price assignment maximizing revenue for which the distance from $r$ to each vertex $v \not\in S$ belongs to the set $\mathcal{D'}$:
$$\mathcal{D'} = \set{0} \cup \longset{b_i}{i \in B_H} \cup \longset{e_i}{i \in S}$$
\end{coll:fixedDepths}

The proof of this corollary follows immediately from the proof of Lemma \ref{lemma:depths}.
Note that all buyers with destination vertices in $S$ generate the same revenue under all feasible price assignments.
The argument for vertices outside $S$ remains the same.

\subsection{The non-skeleton subproblem}
\label{subsect:non_skeleton_subproblem_apx}

Before proceeding to the proofs, let us begin by introducing additional definitions:
\begin{itemize}
\item A price vector $p$ is said to be feasible for the non-skeleton subproblem if it assigns zero to all skeleton edges.
\item For any feasible prices $p$, $\rev_j(p)$ is the revenue generated by buyers from $B_j$ under prices $p$.
\item For a non-skeleton component $C$ and feasible prices $p$ by $\rev_{j, C}(p)$ let us denote the part of $\rev_j(p)$ resulting from selling edges from $C$.
It is well-defined, i.e. is the same regardless of which shortest path between $u_i$ and $v_i$ is assigned to $i$-th buyer for all $i \in B_j$.
For a fragment $F$ from $D_j$ or $D_{j+1}$ let us define $\rev_{j, F}(p)$ in a similar way.
Since each such $F$ is a union of a subgraph of the skeleton and non-skeleton components, $\rev_{j, F}(p)$ is also well-defined. 
\item For any subgraph $S$ of $G$, by $B_{j, S}$ let us denote the set of such buyers $i \in B_j$ that $u_i$ or $v_i$ is not a skeleton vertex and belongs to $S$.
\end{itemize}



\newtheorem*{lemma:randomisedNonSkeleton}{Lemma \ref{lemma:randomisedNonSkeleton}}
\begin{lemma:randomisedNonSkeleton}
Let $p$ be the price vector found by the randomized algorithm and $q$ be any price vector feasible for the non-skeleton subproblem.
Then, the following inequality holds:
$$ \Exp{\rev_{j, H}(p)} \geq \frac{1}{\CrevToOpt} \rev_{j, H}(q) $$
\end{lemma:randomisedNonSkeleton}

\begin{proof}

Consider a non-skeleton component $C$ and a fragment $F$ in $D_{j+1}$ containing $C$.
Assume that $F$ has been painted black and all the other subparts of $H$ have been painted white.
Then, by Remark \ref{rem:splittingBuyers} each buyer $i \in B_{j, C}$ has a budget of $b_i$ for purchasing a path in $C$ from her destination vertex to its skeleton representative because her other non-skeleton section is either empty or in a white fragment of $D_{j+1}$.
Note that this exactly reflects the situation modeled by the rooted instance associated with $C$.
Only buyers from $B_{j,C}$ contribute to $\rev_{j,C}$, so $p$ maximizes $\rev_{j, C}$.
$$ \rev_{j, C}(p) \geq \rev_{j, C}(q) $$ 

In reality, however, the other fragments from $D_{j+1}$ contained in $H$ also can be black.
Nevertheless, each of them is white independently of the color of $F$ with probability $\frac{1}{2}$.
Thus each buyer $i \in B_{j,C}$ can spend up to $b_i$ on her non-skeleton section in $C$ with probability at least $\frac{1}{2}$.
$$ \LExp{\rev_{j,C}(p)}{\text{F is black}} \geq \frac{1}{2} \rev_{j, C}(q) $$
$$ \Exp{\rev_{j, C}(p)} \geq \frac{1}{4} \rev_{j, C}(q) $$

By summing the above inequality over all non-skeleton components in $H$, we obtain the desired lower bound:
$$ \Exp{\rev_{j, H}(p)} \geq \frac{1}{4} \rev_{j, H}(q)$$

\end{proof}

\newtheorem*{coll:derandomisedNonSkeleton}{Corollary \ref{lemma:derandomisedNonSkeleton}}
\begin{coll:derandomisedNonSkeleton}
There exists a deterministic polynomial algorithm which for a non-skeleton subproblem on $j$-th level finds prices achieving at least $\frac{\mname{NSKOPT}_j}{\CrevToOpt}$ revenue.
\end{coll:derandomisedNonSkeleton}

\begin{proof}
Let us begin the proof by derandomizing the procedure from Lemma \ref{lemma:randomisedNonSkeleton} for pricing edges in a single fragment $H \in D_j$.
Note that in at least one of its possible outcomes the computed prices generate at least as must revenue as the expected value.
Each possible outcome is determined by assigning one of two colors for each fragment $F \in D_{j+1}$ contained in $H$.
Since the number of them is $\OO{k}$, there are only $2^{c k}$ possibilities (for a constant $c$).
As $k = \ceil{ \log^{\frac{1}{2}}{m}}$, this number can be bounded by a polynomial in $m$.
The algorithm iterates over all possible color assignments and chooses the solution maximizing revenue in $H$.
Prices $p$ found this way achieve revenue not smaller than the expected revenue achieved by the randomized procedure.
Thus, for any feasible prices $q$:
\begin{equation}
\label{math:fragment}
\rev_{j, H}(p) \geq \frac{1}{4} \rev_{j, H}(q)
\end{equation}

The algorithm performs this derandomized procedure for each fragment of $D_j$ and prices the edges accordingly.
Let us denote the resulting price assignment by $p_{full}$.
Recall that the randomized procedure does not make any assumptions on prices of edges outside the single part $H$.
Thus, Inequality \ref{math:fragment} holds for $p_{full}$, too.
By taking $q = p_{opt}$, an optimal solution to the non-skeleton subproblem on $j$-th level, we obtain the desired upper bound on generated revenue.

$$ \sum_{H \in D_j} \rev_{j, H}(p_{full}) \geq \sum_{H \in D_j} \frac{1}{\CrevToOpt} \rev_{j, H}(p_{opt}) $$
$$ \rev_j(p_{full}) \geq \frac{1}{\CrevToOpt} \rev_j(p_{opt}) = \frac{1}{\CrevToOpt} \mname{NSKOPT}_j$$

\end{proof}

\let\CrevToOpt\undefined

\newcommand{\Crounding}{1024} 
\newcommand{\CrevToOpt}{2048} 
\newcommand{\CmyScoreToMaxScore}{32} 
\newcommand{\CrevToScore}{4} 
\newcommand{\CrevToMyScore}{128} 
\newcommand{\CroundedRevToMyScore}{512} 


\subsection{Decomposing the skeleton}
\label{subsect:decomposeSkeleton_apx}

\newtheorem*{rem:innerOuterPartition}{Remark \ref{rem:innerOuterPartition}}
\begin{rem:innerOuterPartition}
Inner  segments and outer extensions form an edge-disjoint partition of $\mname{SK}_j(F)$.
\end{rem:innerOuterPartition}

\begin{proof}

Let us begin the proof by showing that inner segments form an edge-disjoint partition of $\mname{SK}_j(F) \cap F$, that is all skeleton edges inside $F$.
Since every border vertex is an endpoint of a segment, each segment is fully contained within a certain fragment in the current decomposition.
Thus, as segments form a disjoint cover of $\mname{SK}_j$, the inner segments of the skeleton of $F$ form a disjoint cover of $\mname{SK}_j(F) \cap F$.

By definition, each edge in $\mname{SK}_j(F) \setminus F$ belongs to an outer extension of $\mname{SK}_j(F)$.
Note that every outer extension defines a simple cycle in $\mname{SK}_j(F)$ consisting of two non-empty paths: the outer extension itself and a path in $F$ connecting its endpoints.
Thus, if an edge belonged to two distinct outer extensions, it would lie on two simple cycles, which leads to contradiction.
\end{proof}

\subsection{Price rounding}
\label{subsect:priceRounding_apx}

\newtheorem*{lemma:rounding}{Lemma \ref{lemma:rounding}}
\begin{lemma:rounding}
	\textbf{(rounding)}
	There exists a price assignment obtaining revenue of at least $\frac{\mname{SKOPT}_j}{4}$ such that each segment's length belongs to the following set:
	$$ P = \longset{ \frac{m b_{max}}{2^t} }{ t \in \set{0, 1, \dots, \ceil{\log \left(\Crounding \cdot m^2 \cdot |B_j| \right)} } } \cup \set{0} $$
	Here $b_{max}$ is the greatest budget of buyers in $B_j$ and $m$ is the number of edges in $G$.
	\end{lemma:rounding}
	
	\begin{proof}
	Consider any prices generating revenue $\mname{SKOPT}_s$.
	If a price of any edge is greater than $b_{max}$, it can be lowered to $b_{max}$ without loss of revenue.
	After such a modification, the length of any path is at most $m \cdot b_{max}$.
	We further round the prices down so that the price of each segment belongs to $P$.
	Let us consider a segment $s$ with cost $p_s$ and let $u$, $v$ be its endpoints.
	We choose a $u$-$v$ path of length $p_s$ fully contained within $s$ and set a maximal such $x \in \left[0, 1 \right]$ that $p_s \cdot x \in P$.
	Note that either $x > \frac{1}{2}$ or $x = 0$. The latter implies $p_s < \frac{b_{max}}{\Crounding \cdot m \cdot |B_j|}$.
	Next we multiply the cost of the edges along the chosen path by $x$.
	By applying this procedure to all segments in the graph, we ensure that segments' lengths belong to $P$.
	Let us fix any path, which initially has length $d$.
	After rounding the prices down, its cost is at least $\frac{d}{2} - \frac{b_{max}}{\Crounding \cdot |B_j|}$.
	As a consequence, the cost of a cheapest path desired by a buyer, which initially was equal $d$, will be between $\frac{d}{2} - \frac{b_{max}}{\Crounding \cdot |B_j|}$ and $d$ after the rounding.
	Summing these inequalities for all buyers results in a lower bound on the revenue of $\frac{\mname{SKOPT}_j}{2} - \frac{b_{max}}{\Crounding}$.
	As $\mname{SKOPT}_j \geq b_{max}$, this concludes the proof.
	
	\end{proof}

\subsection{Pricing strategies}
\label{subsect:pricingStrategies_apx}

\subsubsection{Approximating revenue}
\label{subsubsect:approx_rev_apx}

\newtheorem*{lemma:outerOverpricing}{Lemma \ref{lemma:outerOverpricing}}
\begin{lemma:outerOverpricing}
	Consider a valid combination of pricing strategies and denote the strategy for a fragment $F \in D_j$ as $s_F$.
	Then, for any price assignment $p$ implementing that combination of strategies:
	$$ \sum_{F \in D_j} \rev_{j, F, s_F}(p) \geq \frac{b_{max}}{512}
	\quad \implies \quad
	\rev(p) \geq \frac{1}{4} \sum_{F \in D_j} \rev_{j, F, s_F}(p)$$
	\end{lemma:outerOverpricing}

\begin{proof}

Consider a fragment $F \in D_j$ and a buyer $i$ from $B_{j,F}$, who desires paths connecting $u_i$ and $v_i$.
We fix any such path and by $d$, $d'$ denote its length given respectively the prices $p$ or prices $p$ modified in such a way, that the length of each outer extension $o$ of $\mname{SK}_j(F)$ equals $r_{s_F, o}$.
If $l_{s_F, o} = 0$, then $r_{s_F,o} \leq \frac{b_{max}}{\Crounding \cdot m \cdot |B_j|}$.
Thus, due to decreasing the cost of each outer extension $o$ from $r_{s_F, o}$ to its original cost, the path's length can decrease by at most $\frac{b_{max}}{\Crounding |B_j|}$.
Otherwise, $r_{s_F, o} = 2 l_{s_F, o}$ so strategy $s_F$ assumes each outer extension $o$ of $\mname{SK}_j(F)$ to be at most two times longer than it is under prices $p$.
Hence, $d \geq \frac{d'}{2} - \frac{b_{max}}{\Crounding |B_j|}$.
Of course, $d \leq d'$.

Since the above inequalities hold for all $u_i$-$v_i$ paths, it is also true for the distance between $u_i$ and $v_i$.
Hence, every buyer contributing to $\rev_{j, F, s_F}(p)$ also contributes to $\rev_{j,F}(p)$.
By summing the inequality for all those buyers, we obtain:
$$ \rev_{j, F}(p) \geq \frac{\rev_{j, F, s_F}(p)}{2} - |B_{j,F}| \cdot \frac{b_{max}}{\Crounding |B_j|}$$
When applying to all parts of the current decomposition, we have:
$$ \rev_j(p) \geq \frac{1}{2} \sum_{F \in D_j} \rev_{j, F, s_F}(p) - \frac{b_{max}}{\Crounding}$$
As $ \sum_{F \in D_j} \rev_{j, F, s_F}(p) \geq \frac{b_{max}}{512}$, this concludes the proof.

\end{proof}

\subsubsection{Bounding the number of pricing strategies}
\label{subsubsect:strategyBounding_apx}

\newtheorem*{lemma:partSegmentation}{Lemma \ref{lemma:partSegmentation}}
\begin{lemma:partSegmentation}
For each fragment $F \in D_j$, there are $\OO{k}$ inner  segments and outer extensions in its skeleton.
\end{lemma:partSegmentation}

\begin{proof}

Let us consider a graph resulting from converting each inner  segments and outer extension from $\mname{SK}_j(F)$ into a single edge, let us denote it by $H$.
We will prove that $H$ contains $\OO{k}$ vertices, which is sufficient, since $H$, being a cactus, can have at most twice as many edges as vertices.

Like in $\mname{SK}_j(F)$, all edges of $H$ must lie on a path between border vertices on $j$-th level.
If it was not the case, no edge of the corresponding segment would lie on such a path in $\mname{SK}_j(F)$ and hence the whole segment would not belong to the skeleton on $j$-th level.

Consider a tree of biconnected components of $H$ rooted in a vertex denoted $r$.
Each one of its leaves may be either a cycle or a single edge.
Either way, at least one of its non-topmost (relative to $r$) vertices must be a border one.
If it was not the case, edges of that biconnected component would not lie on any simple path connecting border vertices.
Since a vertex can be non-topmost in only one biconnected component, this results in an upper bound of $\OO{k}$ on the number of leaf components in the tree.

Let us define a coloring of biconnected components in $H$.
A biconnected component is said to be black either if it has at least two children in the rooted tree of biconnected components or one of its non-topmost vertices is a border vertex.
Otherwise, it is colored red.
Since the numbers of border vertices and leaf components are bounded, there are only $\OO{k}$ black components.

Note that a red component must be single edge in $H$.
It's because only its topmost vertex and its child's topmost vertex can have degree greater than two.
Any other vertex in a red component, having degree two and not being a border vertex, would be removed during the compression, so would not be an endpoint of a segment.

Let us say that, hypothetically, a red component is a child of another red component.
Then the corresponding biconnected components could be merged using the first operation of the compressing procedure from Definition \ref{def:segment}.
Hence, every red biconnected component can have only black children.
Thus, there can be only as many non-leaf red components as there are black components.
From this follows that there are $\OO{k}$ biconnected components in $H$.

Note that every vertex in $H$ with degree greater than two is a topmost vertex of a biconnected component.
Hence, each vertex in $H$ is either a topmost vertex of a biconnected component or a border vertex.
There are only $\OO{k}$ vertices and thus $\OO{k}$ edges in $H$.
Since each inner or outer extension in $\mname{SK}_j(F)$ corresponds to an edge in $H$, there are only $\OO{k}$ of them in the skeleton of $F$.
	
\end{proof}

\newtheorem{coll:strategiesBound}{Corollary \ref{coll:strategiesBound}}
\begin{coll:strategiesBound}
For each fragment on $j$-th level of decomposition, the number of pricing strategies is polynomial in $m$ and $n$. 
\end{coll:strategiesBound}

\begin{proof}
Let us fix a single fragment $F \in D_j$.
A pricing strategy assigns one of $|P|$ possible values to each inner segment of $\mname{SK}_j(F)$ and one of $|P'|$ intervals to each outer extension.
Note that $|P| \leq |P'|$ and $|P'|$ is $\OO{\log m n}$ ($n = |B|$, so $n \geq |B_j|$).
Thus, there exist such constants $c_1$ and $c_2$ that the number of strategies can be bounded by $\bra{ c_1 \cdot \log m n} ^ { c_2 k}$, which can be bounded by a polynomial in $n$~and~$m$:
$$ \bra{ c_1 \cdot \log m n} ^ { c_2 k} \leq  2 ^ {\bra{ \log c_1 \cdot \log \log m n } \cdot 2 c_2 \sqrt{\log m} } \leq 2^{2 c_2 \log c_1 \cdot \log n m } = \bra{nm} ^ {2 c_2 \log c_1} $$
\end{proof}

\subsection{Solution for a single fragment and a fixed pricing strategy}
\label{subsect:singleFragment_apx}

Before proceeding to the proof of Lemma~\ref{lemma:singleFragment}, we provide a detailed analysis of approximate revenue generated by possible price assignments for an individual segment $S$.

Let $B'$ be any subset of $B_S$.
By $\mname{CONTRIB}_{S, B', s_F}$ let us denote the maximum cost paid by buyers from $B'$ for the edges of $S$ in any envy-free solution based on prices implementing $s_F$.
In other words, $\mname{CONTRIB}_{S, B', s_F}$ is the maximum possible contribution towards $\rev_{j, F, s_F}$ resulting from buyers of $B'$ purchasing edges in $S$.

\begin{lemma}
	\label{lemma:thirdCyclic}
	If the third option was chosen in a cyclic segment $S$, the expected approximate revenue under $s_F$ generated by buyers from $B_{S, l}$ purchasing edges from $S$ is at least $\frac{1}{4} \mname{CONTRIB}_{S, B_{S, l}, s_F}$
	\end{lemma}
	
	\begin{proof}
	Note that a buyer $i \in B_{S, l}$ cannot pay more than $b_{i,l}$ for edges in $S$ under any prices implementing $s_F$.
	Thus, revenue achieved in the artificial rooted instance is not smaller than $\mname{CONTRIB}_{S, B_{S, l}, s_F}$.
	
	With probability of at least one-fourth a buyer $i \in B_{S, l}$ is able to afford to pay $b_{i, l}$ for a $\repr_j(u_i)$-$l$ path, as the first two options in the segment of $\repr_j(v_i)$ (regardless whether it's cyclic or not) guarantee her a free path from $\repr_j(v_i)$ to $v'_i$ or $v''_i$.
	Thus, the expected revenue is at least one-fourth of the revenue from the hypothetical rooted instance, which concludes the proof.
	\end{proof}

The same reasoning can be applied to the fourth solution and buyers from $B_{S,r}$ yielding a lemma analogous to Lemma \ref{lemma:thirdCyclic}.

	\begin{lemma}
		\label{lemma:thirdAcyclic}
		For an acyclic segment $S$, let $B''_S$ be the set of buyers $i \in B_S$ with $\max \set{b_{i,l}, b_{i,r}} \geq \frac{c}{2}$.
		If the third solution was chosen in $S$, the expected approximate revenue under $s_F$ generated in $S$ by buyers from $B''_S$ is at least $\frac{1}{8} \mname{CONTRIB}_{S, B''_S, s_F}$.
		\end{lemma}
		
		\begin{proof}
				
			With probability at least $\frac{1}{4}$, $\dist_{s_F}(v''_i, \repr_j(v_i)) = 0$, the same holds for $\dist_{s_F}(v'_i, \repr_j(v_i))$.
			Thus, for each buyer $i \in B''_s$ with probability at least $\frac{1}{4}$ there exists a $\repr_j(u_i)$-$\repr_j(v_i)$ path of length at most $b_i$.
			In such a case, she would pay $\frac{c}{2}$ for edges in $S$.
			Since $S$ is a path connecting $l$ and $r$, it is impossible for a buyer to pay more than $c$ for a simple path within $S$ under any prices implementing $s_F$.
			Thus, each customer from $B''_s$ with probability at least $\frac{1}{4}$ pays at least half as much for edges in $S$ as they would given any other price assignment.
			Hence, the expected contribution of $B''_S$ towards $\rev_{j, F, s_F}$ generated by purchasing edges from $S$ is at most eight times smaller than its maximal possible value.
			
			\end{proof}

			\begin{lemma}

\label{lemma:fourthAcyclic}
If the fourth option was chosen in an acyclic segment $S$, the expected approximate revenue under $s_F$ generated by buyers from $B'_S$ by purchasing edges from $S$ is at least $\frac{1}{4} \mname{CONTRIB}_{S, B'_S, s_F}$.
\end{lemma}

\begin{proof}

It is sufficient to prove the lemma for a fixed pivot edge because the algorithm iterates over all possibilities.
We will first show that the edge designated to be a pivot edge, denoted $e$, indeed is one.
Then, we will prove that the found prices result in high approximate revenue compared to other solutions, in which no buyer purchases $e$.

Since the budgets of buyers in the rooted instances are smaller than $\frac{c}{2}$, the distances from both $l$ and $r$ to the pivot edge must be too.
Thus, the price of $e$ is positive, and it indeed divides the vertices within the distance of $\frac{c}{2}$ from $l$ and $r$.

No customer $i \in B'_{S,l}$ can pay more than $b_{i,l}$ for edges inside $S$ under any prices implementing $s_F$ as long as she does not buy $e$.
The same holds for $B'_{S,r}$ and $b_{i,r}$.
Thus, in this case the revenue achieved in the two rooted instances by optimal solutions is not smaller than the maximal approximate revenue generated by $B'_S$ in $S$.

Thanks to the first two options, each customer $i \in B'_S$ with probability at least $\frac{1}{4}$ pays nothing for edges in the other segment she is involved in, so she has a budget of $b_{i,l}$ (or $B_{i,r}$) to pay for edges in $S$.
Thus, the expected approximate revenue in $S$ from $B'_S$ is at most four times smaller than the total revenue obtained in the two rooted instances, which concludes the proof.
\end{proof}

\newtheorem*{lemma:singleFragment}{Lemma \ref{lemma:singleFragment}}
\begin{lemma:singleFragment}
	Let $p$ and $q$ be two price assignments implementing $s_F$ such that $q$ maximizes $\rev_{j, F, s_F}$ and $p$ is the result of the randomized algorithm for pricing the skeleton edges of a single fragment.
	Then, the following inequality holds:
	$$ \Exp{ \rev_{j, F, s_F}(p) } \geq \frac{1}{\CmyScoreToMaxScore} \rev_{j, F, s_F}(q)$$
	\end{lemma:singleFragment}

	\begin{proof}

		
		Let us fix any two assignments of paths to the buyers such that each buyer $i \in B_{j,F}$ is assigned a $u_i$-$v_i$ path, which is the cheapest under prices $p$ or $q$ respectively.
		Of course, only as long as its cost does not exceed $b_i$.
		Furthermore, each outer extension $o$ of $\mname{SK}_j(F)$ is assumed to have length equal $r_{s_F, o}$.
		
		By $\inv_{j, S, s_F}(p)$ let us denote the total cost paid by buyers from $B_S$ for edges in the inner segment $S$ and by $\inv_{j, F, s_F}(p)$ its sum over all inner segments of $\mname{SK}_j(F)$.
		The total cost paid by buyers from $B_{j, F}$ in segments they are not involved in is denoted $\full_{j, F, s_F}(p)$.
		Note that this includes the whole revenue generated by $B_{j, F}$ in outer extensions of $\mname{SK}_j(F)$.
		Obviously $\rev_{j, F, s_F}(p) = \inv_{j, F, s_F}(p) + \full_{j, F, s_F}(p)$.
		The same equality holds for $q$, for which both values are defined in the same way.
		We will prove the lemma by separately showing lower bounds for revenue generated by selling parts of segments to involved buyers and selling whole segments to all others.
		
		\textbf{Selling whole segments:}
		Consider any buyer $i \in B_{j,F}$ buying a $\repr_j(u_i)$-$\repr_j(v_i)$ path under the prices $q$.
		By $u'_i$ and $v'_i$ let us denote respectively the first and last endpoint of a segment on this path.
		Obviously, she contributes $\dist_{s_F}(u'_i, v'_i)$ towards $\full_{j, F, s_F}(q)$.
		Let $u''_i$ and $v''_i$ be the other endpoints of segments she is involved in.
		If she is not involved in the given segment, for example because $\repr_j(u_i) = u'_i$, we set $u'_i = u''_i$ (or $v'_i = v''_i$).
		
		Note that the buyer $i$ cannot be 'doubly' involved in a segment.
		In other words, if $\repr_j(u_i)$ and $\repr_j(v_i)$ are not endpoints of segments, they must belong to different ones. 
		Thus, there is such a combination of the first two options for the segments she is involved in, that there exist free $\repr_j(u_i)$-$u'_i$ and $\repr_j(v_i)$-$v'_i$ paths in those segments.
		In that case $\repr_j(v_i)$-$v''_i$ and $\repr_j(u_i)$-$u''_i$ paths in those segments are not cheaper under the prices $p$ than they are under prices $q$.
		Thus, under the prices $p$ a cheapest $\repr_j(u_i)$-$\repr_j(v_i)$ path has length $\dist_{s_F}(u'_i, v'_i)$ and leads through $v'_i$ and $u'_i$.
		Although other shortest $\repr_j(u_i)$-$\repr_j(v_i)$ paths can exist, their cost must also include paying $\dist_{s_F}(u'_i, v'_i)$ for traversing whole segments. 
		Thus, $i$ contributes $\dist_{s_F}(u'_i, v'_i)$ towards $\full_{j, F, s_F}(p)$.
		
		Options for all the inner segments of $\mname{SK}_j(F)$ are chosen independently and each one with probability $\frac{1}{4}$.
		Hence, the situation described above occurs with probability at least $\frac{1}{16}$, which yields the following inequality:
		\begin{equation}
		\label{math:wholeSegments}
		\Exp{ \full_{j, F, s_F}(p) } \geq \frac{1}{16} \full_{j, F, s_F}(q)
		\end{equation}
		
		\textbf{Selling parts of segments:}
		
		Let us consider a single inner segment $S$ of $\mname{SK}_j(F)$ and let $B_S$ be the set of buyers involved in $S$.
		
		If $S$ contains at least one cycle, let us split $B_S$ into two disjoint sets: $B_{S, l}$ and $B_{S, r}$, which are defined in Section \ref{subsubsect:cyclicSegments}.
		If the third option is chosen, Lemma \ref{lemma:thirdCyclic} guarantees that the expected total cost paid by $B_{S, l}$ for edges in $S$ is at most four times smaller than it is under prices $q$.
		If the fourth option is chosen, the same holds for $B_{S, r}$.
		Since both options are chosen with probabilities equal $\frac{1}{4}$, this reasoning yields the following bound:
		\begin{equation}
		\label{math:cyclicSegments}
		\Exp{\inv_{j, S, s_F}(p) \geq \frac{1}{16} \inv_{j, S, s_F}(q)}
		\end{equation}
		
		If $S$ contains no cycles, its prices were constructed using the procedure from Section \ref{subsubsect:acyclicSegments}.
		Let us consider $B'_S$ and $B''_S$ separately.
		By Lemma \ref{lemma:thirdAcyclic}, choosing the third option guarantees the expected approximate revenue from $B'_S$ in $S$ under prices $p$ to be at most eight times smaller than under prices $q$.
		If the fourth option was chosen, by Lemma \ref{lemma:fourthAcyclic} $B''_S$ contributes towards $\Exp{\inv_{j, S, s_F}(p)}$ at least one-fourth of what it contributes towards $inv_{j, S, s_F}(q)$.
		Both options are chosen with probability $\frac{1}{4}$ so this results in the following lower bound:
		\begin{equation}
		\label{math:acyclicSegments}
		\Exp{\inv_{j, S, s_F}(p)} \geq \frac{1}{32} \inv_{j, S, s_F}(q)
		\end{equation}
		
		Applying Inequalities \ref{math:cyclicSegments} and \ref{math:acyclicSegments} to all inner segments of $\mname{SK}_j(F)$ and combining them with \ref{math:wholeSegments} concludes the proof:
		$$ \Exp{ \rev_{j, F, s_F}(p) } \geq \frac{1}{\CmyScoreToMaxScore} \rev_{j, F, s_F}(q)$$
		
		\end{proof}

\subsection{Constructing a global price assignment}
\label{subsect:wholeGraph_apx}

\begin{lemma}
Let $\mathcal{L}$ be the set of possible lengths of any simple path between border vertices under prices satisfying the rounding lemma (\ref{lemma:rounding}).
There are only polynomially many elements of $\mathcal{L}$.
\end{lemma}

\begin{proof}
Any simple path between two border vertices consists of several whole segments.
Thus, its length must be a sum of elements of $P$ from the rounding lemma.
Because all positive elements of $P$ are in the form $\frac{m b_{max}}{2^t}$, $\mathcal{L}$ consists of multiples of $p_{min} = \min P \setminus \set{0} $ not greater than $m^2 b_{max}$.
Since $p_{min} \geq \frac{b_{max}}{2 \cdot \Crounding \cdot m \cdot |B_j|}$, $|\mathcal{L}|$ is bounded by a polynomial.
\end{proof}

\subsection{The lower bound on approximate revenue}
\label{subsect:SummingUp_apx}

\newtheorem*{lemma:maxScoreLowerBound}{Lemma \ref{lemma:maxScoreLowerBound}}
\begin{lemma:maxScoreLowerBound}
	There exists a price assignment $q \in Q_j$ with total approximate revenue at most $\CrevToScore$ times smaller than the maximal revenue obtained by prices in $Q_j$.
	\end{lemma:maxScoreLowerBound}
	
	\begin{proof}
	
	The lemma is proven by constructing the prices $q$.
	Let $q' \in Q_j$ be a price assignment maximizing $\rev_j(q')$.
	Consider any segment $S$ in on $j$-th level.
	We define the main edges of $S$ as the ones lying on any shortest (under prices $q'$) path connecting its endpoints and fully contained inside it.
	By $p_{min}$ let us denote the smallest positive element of $P$ from the rounding lemma (\ref{lemma:rounding}).
	Prices $q$ for edges of $S$ are defined in the following way:
	\begin{itemize}[nosep]
	\item If the length of $S$ is greater than $p_{min}$ (equivalently at least $2 p_{min}$), the prices of main edges of $S$ are two times smaller than in $q'$.
	Otherwise, they are set to zero.
	\item All the other edges in $S$ have the same prices as in $q'$.
	\end{itemize}
	Note that the length of $S$ under prices $q$ belongs to the set $P$ from Lemma \ref{lemma:rounding}.
	
	Since $\rev_j(q') = \sum_{F \in D_j} \rev_{j, F}(q')$, it is enough to show the inequality for a fixed part $F \in D_j$.
	Consider an outer extension $o$ in $\mname{SK}_j(F)$.
	Let $d$ and $d'$ be its lengths under prices $q$ and $q'$ respectively.
	If $d' < 2 p_{min}$, then $d = 0$ and $r_{s_{F, q}, o} = 0$, i.e. the strategy implemented in $F$ assumes that the length of $0$ is exactly zero.
	Otherwise, $l_{s_{F, q'}, o} \geq \frac{d'}{2}$ and $d \leq \frac{d'}{2}$ so $d \leq l_{s_{F, q'}, o}$.
	Hence, $r_{s_{F, q}, o} \leq l_{s_{F, q'}, o} < d'$.
	Thus, in both cases the strategy for $F$ implemented by $q$ assumes the length of $o$ to be at most the actual cost of $o$ under prices $q'$.
	
	Note that using the upper bounds on lengths of outer extensions instead of their actual lengths is the only difference between calculating approximate and actual revenue.
	Thus, the length of each simple path desired by a buyer from $B_{j, F}$ is not more expensive with regard to $\rev_{j, F, s_{F, q}}(q)$ than to $\rev_{j, F}(q')$.
	Hence, a buyer contributing to $\rev_{j, F}(q')$ also contributes to $\rev_{j, F, s_{F, q}}(q)$.
	Now we are going to show that she does not contribute much less.
	
	When calculating $\rev_{j, F, s_{F, q}}(q)$ the cost of each path is never understated in comparison to $\rev_{j, F}(q)$.
	Thus, it is enough to show that under prices $q$ the paths in $\mname{SK}_j(F)$ are not much shorter than under $q'$.
	Note that if the price assigned to an edge by $q$ is zero, then its price in $q'$ could have been at most $p_{min}$.
	Otherwise, it is at most two times cheaper under $q$ then under $q'$.
	Hence, for each $e \in \mname{SK}_j(F)$, $q(e) \geq \frac{q'(e)}{2} - p_{min}$.
	Thus, the distance between any two vertices under prices $q$ is not smaller than $\frac{d'}{2} - m \cdot p_{min}$, where $d'$ is the distance between them under $q'$.
	Since by calculating approximate revenue the algorithm overstates the distances, this results in the following bound on approximate revenue: 
	$$ \rev_{j, F, s_{F, q}} (q) \geq \frac{1}{2} \rev_{j, F}(q') - m \cdot | B_{j,F}| \cdot  p_{min} $$
	$$ \sum_{F \in D_j}  \rev_{j, F, s_{q,F}}(q) \geq \frac{1}{2} \rev_j(q') - m \cdot  |B_j| \cdot \frac{b_{max}}{\Crounding \cdot m \cdot |B_j|} $$
	Since $\rev_j(q') \geq b_{max}$, this concludes the proof.
	\end{proof}

\let\Crounding\undefined
\let\CrevToMyScore\undefined
\let\CrevToOpt\undefined
\let\CmyScoreToMaxScore\undefined
\let\CrevToScore\undefined
\let\CroundedRevToMyScore\undefined

\end{document}